\documentclass[journal]{IEEEtran}
\usepackage[utf8]{inputenc}
\usepackage{balance}
\usepackage{graphicx}
\usepackage{bm}
\usepackage{amsmath}
\usepackage{amsfonts}
\usepackage{amsthm}
\usepackage{amssymb}
\usepackage{color}
\usepackage{mathrsfs}
\usepackage{upgreek}
\usepackage[utf8]{inputenc}
\usepackage{tikz}
\usetikzlibrary{arrows, calc, decorations.markings, positioning}
\usepackage{cite}
\usepackage{booktabs}
\usepackage{multirow}
\usepackage{amsmath}
\usepackage{amsthm}
\usepackage{algorithm}
\usepackage{algorithmic}
\usepackage{hyperref}
\usepackage{midfloat}
\usepackage{tabularx}
\usepackage[flushleft]{threeparttable}
\ifCLASSOPTIONcompsoc
 \usepackage[caption=false,font=normalsize,labelfont=sf,textfont=sf]{subfig}
\else
 \usepackage[caption=false,font=footnotesize]{subfig}
\fi
\newtheorem{thm}{Theorem}
\newtheorem{assumption}{Assumption}
\newtheorem{definition}{Definition}
\newtheorem{lemma}{Lemma}
\newtheorem{corollary}{Corollary}

\newtheorem{proposition}{Proposition}

\newtheorem{conjecture}{Conjecture}

\usepackage{hyperref}					
\hypersetup{colorlinks,
	linkcolor=blue,%
	anchorcolor=blue,
    citecolor=blue}

\begin{document}
\title{\huge{CP-OFDM Achieves Lower Ranging CRB Than Frequency-Spread Waveforms in the Large-Sample Regime}
}
\author{
	{
	Fan Liu,~\IEEEmembership{Senior Member,~IEEE}, Yifeng Xiong,~\IEEEmembership{Member,~IEEE}, Ya-Feng Liu,~\IEEEmembership{Senior Member, IEEE}, \\Jie Yang,~\IEEEmembership{Member,~IEEE}, Christos Masouros,~\IEEEmembership{Fellow, IEEE}, and Shi Jin,~\IEEEmembership{Fellow, IEEE}
\thanks{This work was supported in part by the National Science and Technology Major Projects of China under Grant 2025ZD1302000, and in part by the National Natural Science Foundation of China (NSFC) under Grant 62522107. (\textit{Corresponding authors: Yifeng Xiong; Shi Jin.})}
\thanks{F. Liu and S. Jin are with the National Mobile Communications Research Laboratory, Southeast University, Nanjing 210096, China (email: fan.liu@seu.edu.cn, jinshi@seu.edu.cn).}
\thanks{Y. Xiong is with the Department of Communication Engineering, Beijing University of Posts and Telecommunications, Beijing 100876, China (e-mail: yifengxiong@bupt.edu.cn).}
\thanks{Y.-F. Liu is with the Ministry of Education Key Laboratory of Mathematics and Information Networks, School of Mathematical Sciences, Beijing University of Posts and Telecommunications, Beijing 102206, China (email: yafengliu@bupt.edu.cn).}
\thanks{J. Yang is with the Frontiers Science Center for Mobile Information Communication and Security, Southeast University, Nanjing 210096, China, and also with the Key Laboratory of Measurement and Control of Complex Systems of Engineering, Ministry of Education, Southeast University, Nanjing 210096, China (email: yangjie@seu.edu.cn).}
\thanks{C. Masouros is with the Department of Electrical and Electronic Engineering, University College London, London WC1E 7JE, UK (email: c.masouros@ucl.ac.uk).}
} 
%

}
\maketitle

\begin{abstract}
The inherent randomness of communication symbols creates a fundamental tension in Integrated Sensing and Communications (ISAC). On the one hand, they enable data transmission while allowing sensing to fully reuse communication resources. On the other hand, their randomness induces waveform-dependent fluctuations that directly affect sensing accuracy. This paper investigates a foundational question arising from this tradeoff: \textit{How does the modulation waveform affect the ranging Cram\'er--Rao Bound (CRB) when sensing reuses random data symbols?} We address this question by revealing a structural factorization of the Fisher information matrix (FIM) for joint delay-amplitude estimation, which separates the deterministic Jacobian of the target geometry from the random frequency-domain signal power induced by the data symbols. This structure yields a Jensen-type universal lower bound on the CRB, which is exactly attained by Cyclic Prefix-Orthogonal Frequency Division Multiplexing (CP-OFDM) under Phase Shift Keying (PSK) constellations. For Quadrature Amplitude Modulation (QAM) and broader sub-Gaussian constellations, we develop an asymptotic perturbation analysis of the inverse FIM and prove that, when the number of transmitted symbols $N$ grows large, CP-OFDM achieves a lower ranging CRB than any frequency-spread orthogonal waveform over the almost-sure event where the random FIM is invertible. This superiority is further extended to amplitude estimation and full joint delay-amplitude estimation. We also characterize the local geometry of the stochastic CRB minimization problem over the unitary group. The analysis reveals that CP-OFDM is a stationary point for finite $N$, and its Riemannian Hessian is positive semidefinite for sufficiently large $N$, establishing its asymptotic local optimality. Numerical results confirm that OFDM outperforms representative waveforms including Single-Carrier (SC), Orthogonal Time-Frequency Space (OTFS), and Affine Frequency Division Multiplexing (AFDM).
\end{abstract}

\section{Introduction}

\IEEEPARstart{I}{ntegrated} sensing and communications (ISAC) has emerged as a key enabling technology for future wireless networks, where communication and sensing functionalities are expected to be jointly supported over shared spectrum, hardware, and signaling resources \cite{liu2022integrated}. By allowing communication signals to probe the surrounding environment, ISAC can provide ubiquitous sensing capabilities without relying on dedicated radar waveforms or separate sensing infrastructure. This resource-sharing principle is particularly attractive for scenarios where high-rate data transmission and accurate environmental awareness are simultaneously required, such as vehicular networks, industrial automation, extended reality, and low-altitude networks \cite{zhang2022enabling}.

A prevailing approach in 5G-Advanced (5G-A) oriented ISAC is to perform sensing with deterministic reference signals embedded in the 5G New Radio (NR) frame structure, e.g., Zadoff-Chu (ZC) sequences, $m$-sequences, or even Chirp signals \cite{wei2023signals}. These signals have favorable correlation properties and are mostly compatible with existing communication protocols. However, they occupy only a small fraction of the radio frame, typically on the order of $10\%$~\cite{liu2026jsac_tut}, leaving the majority of communication resources unexploited for sensing. This resource limitation may become a major bottleneck for high-accuracy and high-resolution sensing. It is therefore natural to reuse random data symbols, which occupy the remaining $90\%$ of time-frequency resources, to further improve sensing performance. This leads to challenging waveform design problems for data-bearing ISAC signals. Modern wireless systems employ a variety of orthogonal modulation formats, including SC, OFDM, OTFS~\cite{hadani2017otfs,gaudio2020effectiveness}, and AFDM~\cite{bemani2023afdm}. These waveforms can often be represented under a generic unitary transform framework. Nevertheless, they map the same random data symbols to the time-domain transmitted signal through different orthogonal bases, thereby inducing distinct statistical structures in the sensing observation. Hence, even under the same constellation and average transmit power, the sensing performance of random data-bearing signals may depend strongly on the modulation basis.

From a broader perspective, this modulation dependence connects ISAC waveform design to the more basic question of how to evaluate the sensing performance of communication waveforms. The performance limits of ISAC are naturally tied to the long-standing connection between information and estimation theories \cite{liu2022survey}. Classical rate-distortion theory relates the description rate of a random source to the distortion level of its reconstruction, while the I-MMSE relation links mutual information over Gaussian channels to the minimum mean square error (MMSE) of input estimation~\cite{guo2005mutual}. These results provide two representative examples showing that information transfer and statistical inference are deeply intertwined. Motivated by this connection, ISAC has been studied from both information- and estimation-theoretic viewpoints. From the information-theoretic side, recent works have characterized communication capacity under channel state sensing, channel discrimination, and finite blocklength sensing constraints~\cite{sutivong2005channel,kim2008state,zhang2011joint,kobayashi2018joint,ahmadipour2024information,joudeh2022joint,wu2022joint,Joudeh2026JSAIT,nikbakht2024integrated,chen2025fundamental}. From the estimation-theoretic side, the CRB-rate region was proposed to characterize the tradeoff between communication rate and sensing accuracy over Gaussian channels~\cite{xiong2023fundamental}. This viewpoint has further triggered studies on multi-antenna ISAC transmission design, where spatial resource allocation is optimized under communication and sensing constraints~\cite{liu2018toward,liu2022crb,hua2024mimo,ren2024fundamental}. While these works provide important insights into joint beamforming design and spatial degrees of freedom (DoFs) in ISAC systems, their models usually rely on deterministic probing waveforms or Gaussian signaling, and therefore largely abstract away from the discrete and random nature of practical communication symbols.

Sensing with random communication signals has recently emerged as a distinct line of investigation \cite{liu2026jsac_tut}. Its central premise is that payload symbols constitute a substantial part of the ISAC waveform, and their statistics directly affect the resulting sensing behavior. In this context, random data symbols have been shown to change ambiguity function sidelobes and the visibility of weak targets in payload based sensing~\cite{lu2024random,liu2024ofdm,liu2025iceberg}. In particular, our previous work~\cite{liu2024ofdm} proved that CP-OFDM achieves the lowest average ranging sidelobe among all orthogonal modulation waveforms with CP under QAM/PSK constellations. Subsequent studies further showed that pulse shaping and constellation shaping provide additional DoFs for controlling how symbol randomness appears in the delay-Doppler domain~\cite{liu2025iceberg,du2024reshaping,zhang2026discrete,han2025constellation}. Nevertheless, ambiguity function and sidelobe metrics mainly describe the correlation behavior of matched-filter based sensing. The CRB provides a more direct measure of estimation accuracy, as it lower-bounds the mean square error (MSE) of any unbiased estimator and does not depend on a particular estimation algorithm \cite{cramer1946mathematical,rao1945information,kay1993fundamentals,vantrees2001detection}. Existing sidelobe-based results provide limited information about how random QAM/PSK symbols shape the CRB under a general orthogonal modulation basis. This leaves open a fundamental estimation-theoretic question for ISAC: \textit{Is CP-OFDM still optimal when sensing performance is measured by the CRB rather than by average ranging sidelobes?}

Addressing this question is considerably more challenging than the conventional CRB analysis for deterministic sensing waveforms. For a deterministic probing signal, the FIM is fixed once the waveform and target parameters are specified. For a random data-bearing communication signal, however, the FIM itself varies with the transmitted symbols, and the CRB involves the expectation of its inverse. This average-inverse operation is highly sensitive to the fluctuation pattern of the FIM. Moreover, the previous sidelobe optimality of CP-OFDM cannot be directly translated into CRB optimality. Ranging sidelobes are correlation metrics and are mainly governed by the fourth-order statistics of the constellation \cite{liu2024ofdm}. By contrast, the CRB is a highly nonlinear function of the random FIM and cannot be characterized by a single moment of constellation symbols. Therefore, establishing CRB-based waveform superiority requires a separate estimation-theoretic analysis of how the modulation basis controls the random FIM fluctuation induced by QAM/PSK symbols.

This paper provides a positive answer to the above question. We show that when the number of transmitted symbols $N$ grows large, CP-OFDM achieves a lower ranging CRB than any frequency-spread orthogonal waveform under QAM/PSK and more general sub-Gaussian constellations. In addition, we prove that CP-OFDM is a stationary point of the stochastic CRB minimization problem for finite $N$, and is a local minimizer for sufficiently large $N$. The main contributions are summarized as follows:
\begin{itemize}
    \item We reveal a structural factorization of the FIM for multi-target ranging with random communication symbols. This factorization separates the deterministic target geometry from the random frequency-domain signal power induced by data symbols, and leads to a universal Jensen-type lower bound on the CRB. The bound is exactly attained by CP-OFDM under PSK constellations.
    \item We formalize frequency-spread waveforms through the discrete root-mean-square (RMS) bandwidth of each orthogonal basis vector. Under QAM and more general sub-Gaussian constellations, we prove that as the number of observed samples grows large, CP-OFDM achieves a lower ranging CRB than any frequency-spread orthogonal waveform over the almost-sure event where the random FIM is invertible.
    \item We characterize the local geometry of the stochastic CRB minimization problem over the unitary group. We prove that CP-OFDM is a stationary point under finite number of observations, and further show that its Riemannian Hessian is positive semidefinite in the large-sample regime, which establishes the asymptotic local optimality of CP-OFDM.
\end{itemize}

The remainder of this paper is organized as follows. Sec.~\ref{model_sec} presents the system model. Sec.~\ref{crb_sec} derives the FIM structure and formulates the ranging CRB. Sec.~\ref{lemma_sec} provides the technical lemmas required for the asymptotic analysis. Sec.~\ref{main_sec} establishes the main results of the paper. Finally, Sec.~\ref{conclusion_sec} concludes the paper.

\emph{Notations}: Matrices are denoted by bold uppercase letters, e.g., $\mathbf{U}$; vectors are denoted by bold lowercase letters, e.g., $\mathbf{s}$; and scalars are denoted by normal font, e.g., $N$. The $n$th entry of a vector $\mathbf{s}$ is denoted by $s_n$, and the $(m,n)$th entry of a matrix $\mathbf{A}$ is denoted by $[\mathbf{A}]_{m,n}$. The transpose, Hermitian transpose, and entry-wise complex conjugate are denoted by $(\cdot)^T$, $(\cdot)^H$, and $(\cdot)^\ast$, respectively. The real/imaginary part, trace, expectation, and probability operators are denoted by $\operatorname{Re}\{\cdot\}$, $\operatorname{Im}\{\cdot\}$, $\operatorname{Tr}\{\cdot\}$, $\mathbb{E}\{\cdot\}$, and $\Pr\{\cdot\}$, respectively. The Kronecker product, Hadamard product, and vectorization are denoted by $\otimes$, $\odot$, and $\operatorname{vec}(\cdot)$, respectively. The entry-wise squared magnitude of a matrix $\mathbf{X}$ is written as $|\mathbf{X}|^2=\mathbf{X}\odot\mathbf{X}^\ast$. The spectral norm and Frobenius norm are denoted by $\|\cdot\|_2$ and $\|\cdot\|_F$, respectively. The notation $\operatorname{Diag}(\mathbf{a})$ denotes a diagonal matrix with the entries of $\mathbf{a}$ on its main diagonal, while $\operatorname{Diag}\{a_n\}_{n=1}^{N}$ denotes the diagonal matrix formed by the sequence $\{a_n\}_{n=1}^{N}$. The identity matrix and all-one vector of length $N$ are denoted by $\mathbf{I}_N$ and $\mathbf{1}_N$, respectively. The minimum and maximum eigenvalues of a Hermitian matrix are denoted by $\lambda_{\min}(\cdot)$ and $\lambda_{\max}(\cdot)$. For Hermitian matrices, $\mathbf{A}\succeq \mathbf{B}$ means that $\mathbf{A}-\mathbf{B}$ is positive semidefinite. The standard asymptotic notations $O(\cdot)$, $o(\cdot)$, and $\Theta(\cdot)$ are used throughout the paper. The symbol $\delta_{m,n}$ denotes the Kronecker-Delta function, which takes value $1$ if $m = n$ and $0$ otherwise.

\section{System Model and Preliminaries}\label{model_sec}

Consider a single-antenna monostatic ISAC system with a co-located transmitter and sensing receiver, which transmits a CP-appended communication signal for both data delivery and multi-target ranging. After CP removal, the echo signal reflected by $L$ point-like targets may be modeled as
\begin{equation}\label{sensing_recepition}
    \mathbf{y}=\sum_{\ell=1}^{L}\beta_{\ell}\mathbf{P}_{\tau_{\ell}}\mathbf{U}\mathbf{s}+\mathbf{n},
\end{equation}
where $\mathbf{s}\in\mathbb{C}^{N\times 1}$ is the random data-symbol vector, $\mathbf{U}\in\mathbb{C}^{N\times N}$ is the orthogonal modulation basis that converts data symbols into a time-domain sequence, $\beta_{\ell}\in\mathbb{C}$ and $\tau_{\ell}\in\mathbb{R}$ denote the complex amplitude and delay of the $\ell$th target, respectively, $\mathbf{P}_{\tau_{\ell}}$ is the periodic delay matrix induced by the CP operation, and $\mathbf{n}\sim\mathcal{CN}(\mathbf{0},\sigma^2\mathbf{I}_N)$ is the additive white Gaussian noise. The $L$ targets are assumed to have distinct delays. To proceed, we specify in the following the data symbol model, the modulation basis, and the periodic delay operator in \eqref{sensing_recepition}.

\subsection{Constellation Symbols}

Let $\mathbf{s}=[s_1,s_2,\ldots,s_N]^T$ represent $N$ data symbols independently drawn from a complex constellation alphabet $\mathcal{S}$, with $s_n\ne 0,\;\forall s_n\in\mathcal{S}$. The constellation is assumed to satisfy the following normalization and symmetry conditions.

\begin{assumption}[Unit Power]\label{assumption1}
The constellation has unit average power, i.e.,
\begin{equation}
    \mathbb{E}\left(|s|^2\right)=1,
\end{equation}
where $s$ denotes a generic symbol drawn from $\mathcal{S}$.
\end{assumption}

\begin{assumption}[Second-Order Symmetry]\label{assumption2}
The constellation has zero mean and zero pseudo-variance, i.e.,
\begin{equation}
    \mathbb{E}\left(s\right)=0,\quad \mathbb{E}\left(s^2\right)=0,
\end{equation}
\end{assumption}

\begin{assumption}[Central Symmetry]\label{assumption3}
For every magnitude level $\rho$ with $\Pr(|s|^2=\rho)>0$, the constellation points on the ring
\[
    \mathcal{S}_{\rho}:=\{s\in\mathcal{S}: |s|^2=\rho\}
\]
are centrally symmetric with respect to the origin, i.e., $s\in\mathcal{S}_{\rho}$ implies $-s\in\mathcal{S}_{\rho}$. Under the uniform distribution over the constellation alphabet, this is equivalent to
\begin{equation}
    \mathbb{E}\left(s\left|\left|s\right|^2\right.\right)=0.
\end{equation}
\end{assumption}
These three assumptions hold for standard square QAM and $M$-PSK with $M\geq 4$ after power normalization. The kurtosis of the constellation is defined as
\begin{equation}
    \kappa:=\frac{\mathbb{E}\{|s-\mathbb{E}\{s\}|^4\}}{\mathbb{E}\{|s-\mathbb{E}\{s\}|^2\}^2}=\mathbb{E}\{|s|^4\},
\end{equation}
where the equality follows from Assumptions \ref{assumption1} and \ref{assumption2}. A standard complex Gaussian symbol has $\kappa=2$, while the practical constellations considered here have smaller kurtosis.
\begin{definition}[Sub-Gaussian Constellation]\label{def_1}
A constellation satisfying Assumptions 1 and 2 is called sub-Gaussian if $\kappa<2$.
\end{definition}
\noindent PSK and QAM are both sub-Gaussian under this definition. In particular, PSK has constant modulus and thus $\kappa=1$, whereas QAM has non-constant modulus but still satisfies $1<\kappa<2$. The kurtosis values of typical QAM/PSK constellations are listed in Table~\ref{tab: kurtosis}.

\begin{table}[!t]
\caption{Kurtosis values of typical sub-Gaussian constellations}
\label{tab: kurtosis}
\begin{tabular}{l|c|c|c|c}
\hline
\textbf{Constellation} & PSK     & 16-QAM  & 64-QAM   & 128-QAM  \\ \hline
\textbf{Kurtosis}      & 1       & 1.32    & 1.381    & 1.3427   \\ \hline
\textbf{Constellation} & 256-QAM & 512-QAM & 1024-QAM & 2048-QAM \\ \hline
\textbf{Kurtosis}      & 1.3953  & 1.3506  & 1.3988   & 1.3525   \\ \hline
\end{tabular}
\end{table}

\subsection{Modulation Waveforms}
The transmitted block before CP insertion is generated by
\begin{equation}
    \mathbf{x}=\mathbf{U}\mathbf{s},
\end{equation}
where $\mathbf{U}=[\mathbf{u}_1,\mathbf{u}_2,\ldots,\mathbf{u}_N]\in\mathbb{C}^{N\times N}$ belongs to the unitary group
\begin{equation}
    \mathbb{U}(N):=\left\{\mathbf{U}\in\mathbb{C}^{N\times N}\left|\mathbf{U}\mathbf{U}^H=\mathbf{U}^H\mathbf{U}=\mathbf{I}_N\right.\right\},
\end{equation}
which is also known as a ``waveform'' for communication transmission. Different choices of $\mathbf{U}$ lead to different orthogonal waveforms. Four representative examples are \cite{liu2024ofdm}
\begin{subequations}\label{eq:typical_waveforms}
\begin{align}
    &\mathbf{x}_{\rm SC}:=\mathbf{I}_N\mathbf{s},\\
    &\mathbf{x}_{\rm OFDM}:=\mathbf{F}_N^H\mathbf{s},\\
    &\mathbf{x}_{\rm OTFS}:=\left(\mathbf{F}_{N_1}^H\otimes\mathbf{I}_{N_2}\right)\mathbf{s},\quad N_1N_2=N,\\
    &\mathbf{x}_{\rm AFDM}:=\left(\mathbf{\Lambda}_{c_1}^H\mathbf{F}_N^H\mathbf{\Lambda}_{c_2}^H\right)\mathbf{s}.
\end{align}
\end{subequations}
Here, $\mathbf{F}_N$ denotes the normalized $N$-point DFT matrix with
\begin{equation}
    [\mathbf{F}_N]_{m,n}=\frac{1}{\sqrt{N}}e^{-j2\pi(m-1)(n-1)/N},
\end{equation}
and $\mathbf{\Lambda}_c$ denotes the diagonal chirp matrix
\begin{equation}
    \mathbf{\Lambda}_c:=\operatorname{Diag}\left\{1,e^{-j2\pi c},\ldots,e^{-j2\pi c(N-1)^2}\right\}\in\mathbb{C}^{N\times N}.
\end{equation}
In the AFDM case, $c$ represents either $c_1$ or $c_2$, with $c_1$ satisfying $2Nc_1\in\mathbb{Z}$.

To proceed, we introduce the discrete RMS bandwidth and define frequency-spread waveforms accordingly.
\begin{definition}[Discrete RMS Bandwidth]\label{def_2}
For any unit-norm vector $\mathbf{u}\in\mathbb{C}^{N}$, define its frequency-domain representation as $\widetilde{\mathbf{u}}:=\mathbf{F}_N\mathbf{u}$ and the corresponding discrete power spectrum $p_k:=|\widetilde u_k|^2$ for all $k$, where $\sum_{k} p_k=1$. Let $f_k:=(k-1)/N$ denote the normalized frequency grid. The discrete spectral centroid is defined as $\bar f(\mathbf{u}):=\sum_{k} f_k p_k$, and the discrete RMS bandwidth of $\mathbf{u}$ is defined as
\begin{equation}
B_{\rm rms}(\mathbf{u}):=\sqrt{\sum_{k=1}^N |f_k-\bar f(\mathbf{u})|^2 p_k}.
\end{equation}
\end{definition}
This definition mirrors the continuous RMS bandwidth, which is given by the square root of the second central moment of a normalized power spectral density. Here, $\{p_k\}_{k=1}^N$ forms a discrete distribution over the discrete frequency-domain support $\{f_k\}_{k=1}^N$, and $B_{\rm rms}(\mathbf{u})$ is exactly its standard deviation. In particular, when $\widetilde{\mathbf{u}}$ arises from sampling a sufficiently smooth continuous spectrum, the above discrete RMS bandwidth converges to its continuous counterpart as $N\to\infty$.

\begin{definition}[Frequency-Spread Waveform]\label{def_3}
A unitary modulation matrix $\mathbf{U}\in\mathbb{U}(N)$ is called $\alpha$-frequency-spread if there exists a constant $\alpha >0$ independent of $N$, such that
\begin{equation}
\min_{1\le n\le N} B_{\rm rms}(\mathbf{u}_n)\ge \alpha.
\end{equation}
\end{definition}
The above definition formalizes the intuition that each symbol is inherently distributed over a non-vanishing fraction of the band as $N$ grows. We next compute the corresponding $\alpha$ values for the representative waveforms in \eqref{eq:typical_waveforms} and briefly discuss their implications. 
\subsubsection{OFDM} For $\mathbf{U}=\mathbf{F}_N^H$, the $n$th column is $\mathbf{u}_n=\mathbf{F}_N^H\mathbf{e}_n$, where $\mathbf{e}_n$ denotes the $n$th column of $\mathbf{I}_N$. It follows that $\widetilde{\mathbf{u}}_n=\mathbf{F}_N\mathbf{u}_n=\mathbf{e}_n$, i.e., the frequency-domain power spectrum is supported on a single bin and $\{p_k\}_{k=1}^N$ is a Kronecker delta. Therefore $B_{\rm rms}(\mathbf{u}_n)=0$ for all $n$, which yields
\begin{equation}
\alpha_{\rm OFDM}=0.
\end{equation}
This indicates that OFDM is frequency-localized rather than frequency-spread.
\subsubsection{SC} For $\mathbf{U}=\mathbf{I}_N$, $\widetilde{\mathbf{u}}_n=\mathbf{F}_N\mathbf{e}_n$ has constant-magnitude entries with $p_k=1/N$ for all $k$. The centroid is $\bar f(\mathbf{u}_n)=\frac{1}{N}\sum_{k=1}^N f_k=\frac{N-1}{2N}$, and the discrete second central moment evaluates to
\begin{equation}
\left|B_{\rm rms}^{\rm SC}(\mathbf{u}_n)\right|^2=\frac{1}{N}\sum_{k=1}^N\left(f_k-\frac{N-1}{2N}\right)^2=\frac{N^2-1}{12N^2},
\end{equation}
so that
\begin{equation}
\alpha_{\rm SC}=\sqrt{\frac{N^2-1}{12N^2}}.
\end{equation}
In particular, $\alpha_{\rm SC}\to \frac{1}{2\sqrt{3}}$ as $N\to\infty$. 
\subsubsection{OTFS} For $\mathbf{U}=\mathbf{F}_{N_1}^H\otimes\mathbf{I}_{N_2}$ with $N_1N_2=N$, the $n$th column of $\widetilde{\mathbf{U}}:=\mathbf{F}_N\mathbf{U}$ has exactly $N_2$ non-zero entries with equal magnitudes, which implies that $p_k=1/N_2$ on its support and $p_k=0$ otherwise. Moreover, these non-zero bins are equally spaced on the grid $\{f_k\}_{k=1}^N$, which reduces to a uniform distribution over $N_2$ equally spaced samples, yielding
\begin{equation}
\alpha_{\rm OTFS}=\sqrt{\frac{N_2^2-1}{12N_2^2}}.
\end{equation}
Note that $N_2=1$ reduces OTFS to OFDM, while $\alpha_{\rm SC}\to \frac{1}{2\sqrt{3}}$ as $N_2\to\infty$. 
\subsubsection{AFDM} For $\mathbf{U}=\mathbf{\Lambda}_{c_1}^H\mathbf{F}_N^H\mathbf{\Lambda}_{c_2}^H$, we consider $\widetilde{\mathbf{U}}:=\mathbf{F}_N\mathbf{U}=\mathbf{F}_N\mathbf{\Lambda}_{c_1}^H\mathbf{F}_N^H\mathbf{\Lambda}_{c_2}^H$. Under the constraint $2Nc_1\in\mathbb{Z}$, the magnitude-squared pattern of each column of $\widetilde{\mathbf{U}}$ is periodic in frequency with period $\frac{1}{2c_1}$ and is uniform across the $\frac{1}{2c_1}$ bins within one period. Hence each column has $\frac{1}{2c_1}$ equally weighted spectral components. Consequently,
\begin{equation}
\alpha_{\rm AFDM}=\sqrt{\frac{(\frac{1}{2c_1})^2-1}{12(\frac{1}{2c_1})^2}}=\sqrt{\frac{1-4c_1^2}{12}}.
\end{equation}
When $c_1=1/2$, AFDM reduces to OFDM. Similarly, as $c_1\to 0$, $\alpha_{\rm AFDM}\to \frac{1}{2\sqrt{3}}$. 

Overall, $\alpha>0$ indicates that each basis vector carries a non-vanishing amount of spectral spread, so the resulting orthogonal waveform mixes the data symbols over multiple frequency bins. As will be shown later, this mixing amplifies the symbol-induced fluctuation in the transmitted signal, which provides an intuitive explanation for why OFDM becomes increasingly favorable in the large-$N$ regime.

\subsection{Periodic Delay Operators}
All modulation schemes are assumed to be equipped with a CP whose length is no smaller than the maximum delay spread, which transforms the linear convolution between the signal and channel into its circular counterpart. Therefore, after CP removal, each propagation delay can be represented by a circular shift over the length-$N$ block. For an integer delay $k$, the corresponding periodic time-shift matrix is
\begin{equation}
    \mathbf{P}_k:=\left[\begin{array}{cc}
    \mathbf{0} & \mathbf{I}_{k} \\
    \mathbf{I}_{N-k} & \mathbf{0}
    \end{array}\right],\quad k\in\mathbb{Z},
\end{equation}
with
\begin{equation}
    \mathbf{P}_{-k}=\mathbf{P}_{N-k}=\mathbf{P}_k^T=\left[\begin{array}{cc}
    \mathbf{0} & \mathbf{I}_{N-k} \\
    \mathbf{I}_{k} & \mathbf{0}
    \end{array}\right],\quad k\in\mathbb{Z}.
\end{equation}
Since $\mathbf{P}_k$ is circulant, it admits the diagonalization
\begin{equation}
    \mathbf{P}_k=\mathbf{F}_N^H\operatorname{Diag}\left\{e^{\frac{-j2\pi(n-1)k}{N}}\right\}_{n=1}^{N}\mathbf{F}_N.
\end{equation}
This expression naturally extends the circular-shift operator to fractional delays. Specifically, for $\tau\in\mathbb{R}$, we define
\begin{equation}
    \mathbf{P}_\tau:=\mathbf{F}_N^H\operatorname{Diag}\left\{e^{\frac{-j2\pi(n-1)\tau}{N}}\right\}_{n=1}^{N}\mathbf{F}_N.
\end{equation}
Substituting $\mathbf{x}=\mathbf{U}\mathbf{s}$ and the above periodic delay operator into \eqref{sensing_recepition} gives the observation model used for the subsequent FIM and CRB analysis.

\section{The Structure of Ranging CRB}\label{crb_sec}
We are interested in estimating the unknown delay parameters $\left\{\tau_{\ell}\right\}_{\ell=1}^L$ from the noisy observation $\mathbf{y}\sim\mathcal{CN}\left(\bm{\mu}_{\bm{\theta}},\sigma^2\mathbf{I}\right)$, where
\begin{align}
    \bm{\theta}: =\left\{\tau_{\ell},\operatorname{Re}(\beta_{\ell}),\operatorname{Im}(\beta_{\ell})\right\}_{\ell=1}^L=\left[\bm{\tau},\bm{\beta}_a,\bm{\beta}_b\right]^T\in\mathbb{R}^{3L}
\end{align}
represents the $3L$ unknown parameters, with 
\begin{align}
    &\bm{\tau}: = \left[\tau_1,\tau_2,\ldots,\tau_{L}\right]^T\in\mathbb{R}^{L},\\
    &\bm{\beta}_a: = \left[\operatorname{Re}(\beta_1),\operatorname{Re}(\beta_2),\ldots,\operatorname{Re}(\beta_{L})\right]^T\in\mathbb{R}^{L},\\
    &\bm{\beta}_b: = \left[\operatorname{Im}(\beta_1),\operatorname{Im}(\beta_2),\ldots,\operatorname{Im}(\beta_{L})\right]^T\in\mathbb{R}^{L},
\end{align}
and 
\begin{equation}
    \bm{\mu}_{\bm{\theta}}: = \sum\limits_{\ell = 1}^L\beta_{\ell}\mathbf{P}_{\tau_{\ell}}{\mathbf{x}}.
\end{equation}
Under the Gaussian noise, the $(k,\ell)$-th element of the FIM $\mathbf{J}\in\mathbb{R}^{3L\times 3L}$ is given as \cite{kay1993fundamentals,vantrees2001detection}
\begin{equation}
    J_{k,\ell} = \frac{2}{\sigma^2}\operatorname{Re}\left\{\left(\frac{\partial\bm{\mu}}{\partial\theta_k}\right)^H\left(\frac{\partial\bm{\mu}}{\partial\theta_\ell}\right)\right\}.
\end{equation}
The derivatives may be calculated by
\begin{align}
    \nonumber\frac{\partial\bm{\mu}}{\partial \operatorname{Re}(\beta_\ell)} &= \mathbf{P}_{\tau_{\ell}}\mathbf{x} = \mathbf{F}_N^H\operatorname{Diag}\left(\mathbf{Q}^H\mathbf{s}\right)\mathbf{a}\left(\tau_\ell\right),\\
    \nonumber\frac{\partial\bm{\mu}}{\partial\operatorname{Im}(\beta_\ell)} &= j\mathbf{P}_{\tau_{\ell}}\mathbf{x} = j\mathbf{F}_N^H\operatorname{Diag}\left(\mathbf{Q}^H\mathbf{s}\right)\mathbf{a}\left(\tau_\ell\right).\\
    \frac{\partial\bm{\mu}}{\partial\tau_{\ell}} &= \beta_{\ell}\frac{\partial \mathbf{P}_\tau}{\partial \tau}\mathbf{x}   = \beta_\ell\mathbf{F}_N^H\operatorname{Diag}\left(\mathbf{Q}^H\mathbf{s}\right)\dot{\mathbf{a}}\left(\tau_\ell\right),
\end{align}
where $\mathbf{Q}:=\mathbf{U}^H\mathbf{F}_N^H \in\mathbb{U}\left(N\right)$, and
\begin{equation}
    \mathbf{a}\left(\tau\right) = \left[1,e^{-\frac{j2\pi \tau}{N}},\ldots,e^{-\frac{j2\pi(N-1) \tau}{N}}\right]^T\in\mathbb{C}^N
\end{equation}
is the frequency-domain steering vector, with its derivative expressed as
\begin{equation}
    \dot{\mathbf{a}}\left(\tau\right) = \frac{\partial\mathbf{a}\left(\tau\right) }{\partial\tau} = \textstyle\operatorname{Diag}\left\{-\frac{j2\pi\left(n-1\right)}{N}\right\}_{n=1}^N\mathbf{a}\left(\tau\right).
\end{equation}

For notational convenience, let $\mathbf{a}_\ell: = \mathbf{a}\left(\tau_\ell\right),\;\dot{\mathbf{a}}_\ell: = \dot{\mathbf{a}}\left(\tau_\ell\right)$. We may therefore define
\begin{align}\label{mats_def}
    &\nonumber\mathbf{H}_{\tau} = \left[\beta_1\dot{\mathbf{a}}_1,\beta_2\dot{\mathbf{a}}_2,\ldots,\beta_L\dot{\mathbf{a}}_L\right]\in\mathbb{C}^{N\times L},\\
    &\nonumber\mathbf{H}_{a} = \left[\mathbf{a}_1,\mathbf{a}_2,\ldots,\mathbf{a}_L\right]\in\mathbb{C}^{N\times L},\\
    &\mathbf{H}_{b} = \left[j\mathbf{a}_1,j\mathbf{a}_2,\ldots,j\mathbf{a}_L\right] = j\mathbf{H}_a\in\mathbb{C}^{N\times L},
\end{align}
which yield a compact formulation of the FIM as
\begin{equation}\label{full_FIM}
    \mathbf{J} = \frac{2}{\sigma^2}\operatorname{Re}\left\{\mathbf{H}^H\operatorname{Diag}\left(\left|\mathbf{Q}^H\mathbf{s}\right|^2\right)\mathbf{H}\right\},
\end{equation}
where $\mathbf{H} = \left[\mathbf{H}_{\tau},\mathbf{H}_{a},\mathbf{H}_{b}\right]\in\mathbb{C}^{N\times 3L}$. We may then represent the FIM in the following block structure:
\begin{equation}\label{J_partition}
    {\mathbf{J}} = \left[ {\begin{array}{*{20}{c}}
  {{{\mathbf{J}}_{\tau \tau }}}&{{{\mathbf{J}}_{\tau a}}}&{{{\mathbf{J}}_{\tau b}}} \\ 
  {{{\mathbf{J}}_{\tau a}^T}}&{{{\mathbf{J}}_{aa}}}&{{{\mathbf{J}}_{ab}}} \\ 
  {{{\mathbf{J}}_{\tau b}^T}}&{{{\mathbf{J}}_{ab}^T}}&{{{\mathbf{J}}_{bb}}} 
\end{array}} \right],
\end{equation}
in which we have 
\begin{align}
    &\nonumber\mathbf{J}_{\tau \tau} = \frac{2}{\sigma^2}\operatorname{Re}\left\{\mathbf{H}_\tau^H\operatorname{Diag}\left(\left|\mathbf{Q}^H\mathbf{s}\right|^2\right)\mathbf{H}_\tau\right\}\in\mathbb{R}^{L\times L},\\
    &\nonumber\mathbf{J}_{aa} = \mathbf{J}_{bb} = \frac{2}{\sigma^2}\operatorname{Re}\left\{\mathbf{H}_a^H\operatorname{Diag}\left(\left|\mathbf{Q}^H\mathbf{s}\right|^2\right)\mathbf{H}_a\right\}\in\mathbb{R}^{L\times L},\\
    &\nonumber\mathbf{J}_{\tau a} = \frac{2}{\sigma^2}\operatorname{Re}\left\{\mathbf{H}_\tau^H\operatorname{Diag}\left(\left|\mathbf{Q}^H\mathbf{s}\right|^2\right)\mathbf{H}_a\right\}\in\mathbb{R}^{L\times L},\\
    &\nonumber\mathbf{J}_{\tau b} = \frac{2}{\sigma^2}\operatorname{Re}\left\{\mathbf{H}_\tau^H\operatorname{Diag}\left(\left|\mathbf{Q}^H\mathbf{s}\right|^2\right)\mathbf{H}_b\right\}\in\mathbb{R}^{L\times L},\\
    &\mathbf{J}_{ab} = \frac{2}{\sigma^2}\operatorname{Re}\left\{\mathbf{H}_a^H\operatorname{Diag}\left(\left|\mathbf{Q}^H\mathbf{s}\right|^2\right)\mathbf{H}_b\right\}\in\mathbb{R}^{L\times L}.
\end{align}
To ensure the invertibility of $\mathbf{J}$, we assume $N>3L$ without loss of generality, which suggests that the number of independent observations must be greater than that of the parameters to be estimated.

\textit{Remark 1:} The factorization in \eqref{full_FIM} admits a simple interpretation from the viewpoint of deterministic reparameterization of the Fisher information. The matrix $\mathbf{H}$ should be understood as the deterministic Jacobian that maps infinitesimal changes of the physical parameter vector $\boldsymbol{\theta}$ to changes of the frequency-domain CSI. For a fixed transmitted symbol vector, the frequency-domain observation assigns the following diagonal information weight to the CSI:
\begin{equation}
    \mathbf{J}_{\mathrm{CSI}}
    := \frac{2}{\sigma^2}\operatorname{Diag}\left(\left|\mathbf{Q}^H\mathbf{s}\right|^2\right)
    \in\mathbb{R}^{N\times N}.
\end{equation}
Hence, by transforming this CSI-domain Fisher information back to the physical parameter domain, we obtain \eqref{full_FIM}. This viewpoint separates the deterministic target-geometry Jacobian from the FIM of the CSI induced by random symbols, and will be useful for the subsequent optimality analysis.

Let $\hat{\bm{\tau}}$ denote the estimator of the delay vector $\bm{\tau}$. By the Cram\'er--Rao inequality, for a given ISAC signal vector $\mathbf{x}$, the conditional error covariance satisfies
\begin{equation}
    \mathbb{E}\left[{\left(\bm{\tau} - \hat{\bm{\tau}}\right)}{\left(\bm{\tau} - \hat{\bm{\tau}}\right)}^T\left| {\mathbf{x}} \right.\right]\succeq \mathbf{T}_{\tau}{{\mathbf{J}}^{-1}}\left(\mathbf{x}\right)\mathbf{T}_{\tau}^T,
\end{equation}
where
\begin{equation}
    \mathbf{T}_{\tau}:=\left[\mathbf{I}_L,\mathbf{0}_{L\times 2L}\right]\in\mathbb{R}^{L\times 3L}
\end{equation}
is a selection matrix such that $\mathbf{T}_{\tau}\bm{\theta} = \bm{\tau}$. It follows that
\begin{equation}
   \mathbb{E}\left(\left\|\left(\bm{\tau} - \hat{\bm{\tau}}\right)\right\|^2\left|{\mathbf{x}}\right.\right)\ge \operatorname{Tr}\left\{\mathbf{T}{{\mathbf{J}}^{-1}}\left(\mathbf{x}\right)\right\}:=f\left(\mathbf{U},\mathbf{s}\right),
\end{equation}
where $\mathbf{T}:= \mathbf{T}_{\tau}^T\mathbf{T}_{\tau}$. Accordingly, the resulting MSE is bounded by
\begin{align}\label{MCB}
    \nonumber\operatorname{mse}\left(\bm{\tau},\hat{\bm{\tau}}\right) &= \mathbb{E}\left\{\mathbb{E}\left(\left\|\left(\bm{\tau} - \hat{\bm{\tau}}\right)\right\|^2\left|{\mathbf{x}}\right.\right)\right\}\ge \mathbb{E}\left[f\left(\mathbf{U},\mathbf{s}\right)\right]
    \\&=\operatorname{Tr}\left\{\mathbf{T}\mathbb{E}\left[{{\mathbf{J}}^{-1}}\left(\mathbf{U},\mathbf{s}\right)\right]\right\}.
\end{align}
The lower bound in \eqref{MCB}, obtained by averaging the conditional CRB over the random symbol vector $\mathbf{s}$, is commonly referred to as the \textit{Miller–Chang Bound}~\cite{miller1978modified,gini2000cramer}. For simplicity and without ambiguity, we shall refer to \eqref{MCB} as the CRB throughout the remainder of the paper.

Our objective is to characterize how the choice of orthogonal modulation waveform affects the ranging CRB when sensing reuses random communication payloads. We first establish a fundamental lower bound on the CRB that is independent of the modulation basis and constellation, which serves as a universal benchmark for ranging with data-bearing signals. This bound is exactly attained by OFDM under PSK constellations. For QAM and more general sub-Gaussian constellations, we then show that OFDM achieves a smaller CRB than any orthogonal waveform satisfying the $\alpha$-frequency-spread condition for large $N$, with a non-vanishing second-order gap.

\begin{proposition}\label{CRB_Jensen_bound}
The CRB is lower-bounded by
\begin{equation}\label{Jensen_bound_prop}
    \operatorname{Tr}\left\{\mathbf{T}\mathbb{E}\left[{{\mathbf{J}}^{-1}}\left(\mathbf{U},\mathbf{s}\right)\right]\right\} \ge \operatorname{Tr}\left(\mathbf{T}\bar{\mathbf{J}}^{-1}\right),
\end{equation}
where 
\begin{equation}
    \bar{\mathbf{J}} = \mathbb{E}\left({\mathbf{J}}\right) =\frac{2}{\sigma^2}\operatorname{Re}\left(\mathbf{H}^H\mathbf{H}\right).
\end{equation}
Equality holds if and only if $\mathbb{E}(\mathbf{J})=\mathbf{J}$, in which case $\mathbf{J}$ is deterministic.
\end{proposition}
\begin{proof}
    By Jensen's inequality
    \begin{equation}\label{Jensen_bound}
        \mathbb{E}\left({{\mathbf{J}}^{-1}}\right) \succeq \left\{\mathbb{E}\left({\mathbf{J}}\right)\right\}^{-1}=\bar{\mathbf{J}}^{-1}.
    \end{equation}
    Since $\mathbb{E}\left(\left|\mathbf{Q}^H\mathbf{s}\right|^2\right) = \mathbf{1}_N$, we have 
    \begin{align}\label{bar_J}
        \nonumber\bar{\mathbf{J}} &= \frac{2}{\sigma^2}\operatorname{Re}\left\{\mathbf{H}^H\operatorname{Diag}\left[\mathbb{E}\left(\left|\mathbf{Q}^H\mathbf{s}\right|^2\right)\right]\mathbf{H}\right\}\\
        & = \frac{2}{\sigma^2}\operatorname{Re}\left(\mathbf{H}^H\mathbf{H}\right).
    \end{align}
    which yields the stated expression.
\end{proof}
In general, the Jensen bound becomes tight as $N\to\infty$. However, it can also be achieved for finite $N$ under certain signaling schemes, as shown next.
\begin{corollary}\label{CRB_OFDM_PSK}
    The Jensen lower bound is exactly achieved by OFDM modulation with PSK constellations.
\end{corollary}
\begin{proof}
    Let $\mathbf{U}=\mathbf{F}_N^H$. Then $\mathbf{Q}=\mathbf{I}_N$. Consequently, for PSK symbols with $\left|s_n\right|=1$ for all $n$, we have
    \begin{align}
    \left|\mathbf{Q}^H\mathbf{s}\right|^2 = \left|\mathbf{s}\right|^2 = \mathbf{1}_N,
    \end{align}
    which is deterministic and leads to the Jensen bound in \eqref{Jensen_bound}.
\end{proof}

\section{Technical Lemmas}\label{lemma_sec}

While Corollary \ref{CRB_OFDM_PSK} shows that OFDM attains the minimum CRB under PSK constellations, the QAM case is substantially more challenging due to amplitude fluctuations of the data symbols. We therefore analyze QAM through two complementary routes: 1) Comparing OFDM with $\alpha$-frequency-spread orthogonal waveforms via the leading-order CRB gap, and 2) characterizing the asymptotic local behavior of the CRB around the OFDM basis. The following lemmas provide the key ingredients for this analysis.

\begin{lemma}[Non-degenerate Channel Condition]\label{non_degenerate_lemma}
The expected FIM $\bar{\mathbf{J}}$ is invertible if
\begin{equation}\label{non_overlapping_targets}
\tau_k \ne \tau_\ell,\quad \forall\;k \ne \ell,\quad \beta_\ell \ne 0,\quad \forall \ell.
\end{equation}
\end{lemma}
\begin{proof}
   Let us consider an arbitrary size-$3L$ real vector $\mathbf{d} = \left[\mathbf{d}_{\tau}^T,\mathbf{d}_a^T,\mathbf{d}_b^T\right]^T\in\mathbb{R}^{3L}$, where $\mathbf{d}_{\tau},\mathbf{d}_a,\mathbf{d}_b\in\mathbb{R}^L$. Recall \eqref{bar_J}. To establish the invertibility of $\bar{\mathbf{J}}$, it suffices to show that $\mathbf{H}$ is of full rank, namely,
    \begin{equation}
        \mathbf{H}\mathbf{d} = \mathbf{0}\Longrightarrow \mathbf{d} = \mathbf{0}.
    \end{equation}
Define
    \begin{align}
        \mathbf{d}_1 = \mathbf{d}_a + j\mathbf{d}_b,\quad\mathbf{d}_2 = \textstyle-\frac{j2\pi}{N}\operatorname{Diag}\left\{\beta_\ell\right\}_{\ell=1}^L\mathbf{d}_{\tau}.
    \end{align}
Then, we can write
    \begin{align}
         &\mathbf{H}\mathbf{d} = \mathbf{H}_a\mathbf{d}_1+\operatorname{Diag}\left\{n-1\right\}_{n=1}^N\mathbf{H}_a\mathbf{d}_2= \mathbf{H}^{\prime}\left[ \begin{gathered}
    {{\mathbf{d}}_1} \hfill \\
    {{\mathbf{d}}_2} \hfill \\ 
    \end{gathered}  \right],
    \end{align}
where $\mathbf{H}^{\prime}:=\left[ {{{\mathbf{H}}_a},\operatorname{Diag} \left\{ {n - 1} \right\}_{n = 1}^N{{\mathbf{H}}_a}} \right]$, which exhibits a confluent Vandermonde structure as long as \eqref{non_overlapping_targets} holds, and thus has a rank of $2L$~\cite{gautschi1962inverses}. Therefore, if $\mathbf{H}^{\prime}[{{\mathbf{d}}_1},{{\mathbf{d}}_2}]^T=\mathbf{0}$, we have $\mathbf{d}_1=\mathbf{d}_2=\mathbf{0}$, implying $\mathbf{d} = \mathbf{0}$. This completes the proof.
\end{proof}
In the multi-target ranging scenario, the condition \eqref{non_overlapping_targets} always holds, which suggests that  $\bar{\mathbf{J}}$ is always invertible.
\begin{lemma}\label{EFIM_order}
    There exist constants $c_0,c_1>0$, such that
    \begin{equation}
         c_0 N\le\lambda_{\min}\left({\bar {\mathbf{J}}}\right) \le\lambda_{\max}\left({\bar {\mathbf{J}}}\right) \le c_1 N,
    \end{equation}
    for sufficiently large $N$. Accordingly,
    \begin{equation}
        \left\|{\bar {\mathbf{J}}^{-1}}\right\|_2 = \Theta\left(N^{-1}\right).
    \end{equation}
\end{lemma}
\begin{proof}
    See Appendix \ref{proof_lemma_2}.
\end{proof}

Define $\bm{\Delta} = \mathbf{J} - {{{\bar {\mathbf{J}} }}}$, and $\mathbf{Z} := \bar{\mathbf{J}}^{-1} \bm{\Delta}$. We may express the inverse of $\mathbf{J}$ via a resolvent-type expansion as
\begin{align}\label{J_expansion}
    \nonumber\mathbf{J}^{-1}&=\left({{{\bar {\mathbf{J}} }}}+\bm{\Delta}\right)^{-1} = {{{\bar {\mathbf{J}} }^{-1}}}\left(\mathbf{I}+{{{\bar {\mathbf{J}} }^{-1}}}\bm{\Delta}\right)^{-1}
    \\&={{{\bar {\mathbf{J}} }^{-1}}}-\mathbf{R}_1+\mathbf{R}_2-\mathbf{R}_3,
\end{align}
where
\begin{subequations}
\begin{align}
    &\mathbf{R}_1 ={{{\bar {\mathbf{J}} }^{-1}}}\bm{\Delta}{{{\bar {\mathbf{J}} }^{-1}}} = \mathbf{Z}{{{\bar {\mathbf{J}} }^{-1}}}\\
    &\mathbf{R}_2 = {{{\bar {\mathbf{J}} } ^{-1}}}\bm{\Delta} {{{\bar {\mathbf{J}} } ^{-1}}}\bm{\Delta} {{{\bar {\mathbf{J}} } ^{-1}}} = \mathbf{Z} ^2{{{\bar {\mathbf{J}} } ^{-1}}},\\
    &\mathbf{R}_3 = {\bar {\mathbf{J}}  ^{ - 1}\bm{\Delta} }\left(\mathbf{I}+{\bar {\mathbf{J}}  ^{ - 1}\bm{\Delta} }\right)^{-1}\mathbf{R}_2 = \mathbf{Z} ^3\left(\mathbf{I}+\mathbf{Z} \right)^{-1}{{{\bar {\mathbf{J}} } ^{-1}}}\label{R_3_original_formulation}
\end{align}    
\end{subequations}
denote the first-, second-, and higher-order residual terms, respectively. 

Let us further introduce a useful concentration inequality in terms of the spectral norm of $\mathbf{Z}$ through the following lemma.
 \begin{lemma}\label{Z_HW_concerntration}
    For any $\varepsilon > 0$, it holds that
    \begin{equation}\label{concentration_sub_Gaussian}
        \Pr\left\{ \|\mathbf{Z}\|_2 \ge \varepsilon \right\} \le 18L^2 \exp\left\{ -cN \min(\varepsilon^2, \varepsilon) \right\},
    \end{equation}
    where $c > 0$ is a constant independent of $N$. Consequently, for any $k \ge 1$, the $k$-th moment of the spectral norm satisfies
    \begin{equation}
        \mathbb{E}\left( \|\mathbf{Z}\|_2^k \right) = O\left(N^{-k/2}\right).
    \end{equation}
\end{lemma}
\begin{proof}
    See Appendix \ref{proof_lemma_3}.    
\end{proof}

The following proposition holds as a direct consequence of Lemma \ref{Z_HW_concerntration}.

\begin{proposition}
     Let $\Omega_N = \left\{\mathbf{s}\left|\left\|\mathbf{Z}\right\|_2<\rho\right.\right\}$ denote the set where $\mathbf{J}$ is invertible via the Neumann series, with $\rho\in\left(0,1\right)$ being a constant independent of $N$. The event $\Omega_N$ occurs for all sufficiently large $N$ with probability 1. 
\end{proposition}

\begin{proof}
    The complement event of $\Omega_N$ is defined as
    \begin{equation}
        \Omega_N^c:= \left\{\mathbf{s}\left| \|\mathbf{Z}\|_2 \ge \rho\right. \right\}.
    \end{equation}
    According to Lemma 3, it satisfies
    \begin{equation}
        \Pr\{\Omega_N^c\} = \Pr\left\{ \|\mathbf{Z}\|_2 \ge \rho \right\} \le 18L^2 \exp\left( -c\rho^2 N \right).
    \end{equation}
    Since $\Pr\{\Omega_N^c\}$ decays exponentially with $N$, the Borel-Cantelli lemma~\cite{durrett2019probability} implies that the event sequence $\Omega_N$ occurs eventually almost surely. 
\end{proof}

\textit{Remark 2:} The FIM is invertible on the event $\Omega_N$. To see this, note that
\begin{equation}
    \mathbf{J}=\bar{\mathbf{J}}+\bm{\Delta}=\bar{\mathbf{J}}(\mathbf{I}+\bar{\mathbf{J}}^{-1}\bm{\Delta})=\bar{\mathbf{J}}(\mathbf{I}+\mathbf{Z}).
\end{equation}
Since $\bar{\mathbf{J}}$ is invertible and $\|\mathbf{Z}\|_2<\rho<1$ on $\Omega_N$, the matrix $(\mathbf{I}+\mathbf{Z})$ is also invertible. Hence, $\mathbf{J}$ is invertible on $\Omega_N$. By Proposition 2, $\Omega_N$ occurs eventually almost surely as $N\to\infty$, which implies that $\mathbf{J}$ is invertible almost surely for all sufficiently large $N$.


Finally, we present the core lemma for establishing the asymptotic superiority of OFDM, which characterizes the fluctuation of FIM incurred by the random communication symbols. Let $\mathbf{v} = \operatorname{vec}\left({{\mathbf{J}}}\right)\in\mathbb{C}^{9L^2}$. The covariance matrix of $\bm{\Delta}$ is given as
\begin{align}
\mathbb{E}\left[\operatorname{vec}\left(\bm{\Delta}\right)\operatorname{vec}\left(\bm{\Delta}\right)^H\right] = \operatorname{Cov}\left(\mathbf{v},\mathbf{v}\right)=\bm{\Sigma}_{\mathbf{v}},
\end{align}
where
\begin{equation}
    \bm{\Sigma}_{\mathbf{x}}:= \mathbb{E}\left(\mathbf{x}\mathbf{x}^H\right) - \mathbb{E}\left(\mathbf{x}\right)\mathbb{E}\left(\mathbf{x}^H\right)
\end{equation}
is the covariance matrix of the random vector $\mathbf{x}$. The following lemma establishes the optimality of OFDM in minimizing the covariance matrix of $\bm{\Delta}$.

\begin{lemma}\label{bistochastic_prop}
    For sub-Gaussian constellations ($\kappa < 2$),
    \begin{equation}
        \bm{\Sigma}_{\mathbf{v}}\left(\mathbf{U}\right)\succeq \bm{\Sigma}_{\mathbf{v}}\left(\mathbf{F}_N^H\right),\quad \forall\;\mathbf{U}\in\mathbb{U}\left(N\right).
    \end{equation}
\end{lemma}
\begin{proof}
    See Appendix \ref{proof_prop_bistochastic}.
\end{proof}

\section{Main Results}\label{main_sec}
\subsection{OFDM Outperforms Frequency-Spread Waveforms}
We may now establish the superiority of OFDM over frequency-spread waveforms. We first consider the possibility that the FIM $\mathbf{J}$ becomes singular, which can occur on the rare event $\Omega_N^c$. For example, when $\mathbf{U}=\mathbf{I}_N$ (SC modulation), we have $\mathbf{Q}=\mathbf{F}_N^H$. If $\mathbf{s}=s_1\mathbf{1}$ for some QAM/PSK symbol $s_1$, then $\mathbf{Q}^H\mathbf{s}=\sqrt{N}\left[s_1,0,\ldots,0\right]^T$, which forces $\mathbf{J}$ to be singular. The probability of this event equals $|\mathcal{S}|^{-(N-1)}$ and hence decays exponentially with $N$, where $|\mathcal{S}|$ is the cardinality of $\mathcal{S}$. It is also clear that such a singularity never occurs under OFDM modulation. Indeed, when $\mathbf{U}=\mathbf{F}_N^H$, we have $\mathbf{Q}^H\mathbf{s}=\mathbf{s}$, and as long as the constellation alphabet contains no zero elements (which holds for standard QAM/PSK constellations), $\operatorname{Diag}\!\left(\mathbf{Q}^H\mathbf{s}\right)$ is always full rank. Together with the non-degenerate channel condition ensuring that $\mathbf{H}$ is full rank, this implies that $\mathbf{J}$ is invertible for all realizations of $\mathbf{s}$ under OFDM modulation, whereas other modulation schemes may admit singular realizations. 

If the unconditional expected CRB were used to compare the ranging performance of different waveforms, any waveform admitting singular FIM realizations would have an infinite objective value, since $\mathbf{J}^{-1}$ does not exist on such events, while the OFDM objective always remains finite. The comparison would then become trivial and uninformative. To focus on the non-trivial regime in which all considered waveforms yield a well-defined CRB with overwhelmingly high probability, we henceforth restrict the analysis to the almost-sure event $\Omega_N$.


\begin{thm}[Large-$N$ Superiority of OFDM over the Almost-Sure Set]\label{asm_opt_thm}
 Consider uniform constellations satisfying Assumptions \ref{assumption1}--\ref{assumption3} and containing no zero element. Assume that the non-degenerate channel condition \eqref{non_overlapping_targets} holds, and restrict to the almost-sure event $\Omega_N$. Fix any $\alpha>0$, and let $\mathbf{U}\in\mathbb{U}(N)$ be an $\alpha$-frequency-spread waveform in the sense of Definition~\ref{def_3}. Then there exist constants $c_{\alpha}(\mathbf{U})>0$ and $N_0(\mathbf{U})$, both independent of $N$, such that for all $N\ge N_0(\mathbf{U})$,
\begin{equation}
\mathbb{E}\!\left[f(\mathbf{U},\mathbf{s})\left|\Omega_N\right.\right]-\mathbb{E}\!\left[f\!\left(\mathbf{F}_N^H,\mathbf{s}\right)\left|\Omega_N\right.\right]\ge \frac{c_{\alpha}(\mathbf{U})}{N^{2}}.
\end{equation}
Consequently, CP-OFDM achieves a strictly lower ranging CRB than any $\alpha$-frequency-spread waveform in the large-$N$ regime over the almost-sure event $\Omega_N$.
\end{thm}

\begin{proof}
    See Appendix \ref{proof_theorem_1}.    
\end{proof}

\begin{corollary}
Under the same assumptions as in Theorem \ref{asm_opt_thm}, OFDM asymptotically achieves a lower CRB for amplitude estimation than any $\alpha$-frequency-spread waveform over the almost-sure event $\Omega_N$. 
\end{corollary}
\begin{proof}
Let $\mathbf{T}_{\beta}:=\left[\mathbf{0}_{2L\times L},\mathbf{I}_{2L}\right]\in\mathbb{R}^{2L\times 3L}$, which selects the amplitude parameter vector $\boldsymbol{\beta}:=[\boldsymbol{\beta}_a^T,\boldsymbol{\beta}_b^T]^T$, and define $\mathbf{T}:=\mathbf{T}_{\beta}^T\mathbf{T}_{\beta}$. The proof follows the same steps as that of Theorem \ref{asm_opt_thm}, since the choice of the estimated parameter subset affects the objective only through the fixed weighting matrix $\mathbf{T}$ in $\operatorname{Tr}(\mathbf{T}\mathbf{J}^{-1})$. Replacing $\mathbf{T}_{\tau}$ with $\mathbf{T}_{\beta}$ leaves analysis unchanged, and thus yields the same asymptotic superiority conclusion for amplitude estimation.
\end{proof}
\begin{corollary}
Under the same assumptions as in Theorem \ref{asm_opt_thm}, OFDM asymptotically achieves a lower full CRB for joint delay and amplitude estimation than any $\alpha$-frequency-spread waveform over the almost-sure event $\Omega_N$. 
\end{corollary}
\begin{proof}
This follows from Theorem \ref{asm_opt_thm} by letting $\mathbf{T}=\mathbf{I}_{3L}$ in $\operatorname{Tr}(\mathbf{T}\mathbf{J}^{-1})$. Since the proof of Theorem \ref{asm_opt_thm} depends on the parameter subset only through the fixed weighting matrix $\mathbf{T}$, all arguments on $\Omega_N$ remain unchanged, and the same second-order gap analysis applies to the full CRB, proving the stated asymptotic superiority over $\alpha$-frequency-spread waveforms.
\end{proof}

\subsection{OFDM is Locally Optimal in the Large-Sample Regime}
To further substantiate the optimality of OFDM in minimizing the ranging CRB, we investigate the local geometry of the following stochastic optimization problem:
\begin{equation}\label{CRB_minimization}
    \min_{\mathbf{U}\in\mathbb{U}(N)}\;\mathbb{E}\left[f(\mathbf{U})\right].
\end{equation}
In particular, we show that $\mathbf{F}_N^H$ is a stationary point of \eqref{CRB_minimization} and that its Riemannian Hessian is asymptotically positive semidefinite. Together, these properties provide an asymptotic local-optimality certificate for OFDM. 
\begin{thm}[Stationarity of OFDM]\label{stationarity_thm}
Let $\mathbf{K}$ be any skew-Hermitian matrix with $\|\mathbf{K}\|_2=1$ and consider the unit-speed geodesic $\mathbf{U}(t)=\mathbf{F}_N^H e^{t\mathbf{K}}$ over $\mathbb{U}(N)$. For uniform constellations satisfying Assumptions \ref{assumption1}--\ref{assumption3} and containing no zero element, the Riemannian gradient of $\mathbb{E}\left[f(\mathbf{U})\right]$ vanishes at $\mathbf{U}(0)=\mathbf{F}_N^H$, i.e.,
\begin{equation}
\left.\frac{d}{dt}\mathbb{E}\left\{f\left[\mathbf{U}(t)\right]\right\}\right|_{t=0} = 0.
\end{equation}
\end{thm}
\begin{proof}
See Appendix \ref{proof_thm_2}.
\end{proof}
Note that Theorem \ref{stationarity_thm} holds for the finite-$N$ regime and applies to the total expectation. We next establish an asymptotic second-order property at $\mathbf{U}=\mathbf{F}_N^H$, which again holds for the total expectation rather than restricting to $\Omega_N$. 
\begin{thm}[Semidefiniteness of the Large-$N$ Hessian at OFDM]\label{hessian_thm}
Under the same assumptions as in Theorem~\ref{stationarity_thm}, let
$\mathbf K$ be any skew-Hermitian matrix with $\|\mathbf K\|_2=1$, and consider
the geodesic $\mathbf U(t)=\mathbf F_N^H e^{t\mathbf K}$. Then there exists
$N_0(\mathbf K)$ such that, for all $N\ge N_0(\mathbf K)$,
\begin{equation}
\left.\frac{d^{2}}{dt^{2}}\mathbb{E}\left\{f\left[\mathbf{U}(t)\right]\right\}\right|_{t=0}\ge 0.
\end{equation}
\end{thm}
\begin{proof}
See Appendix \ref{proof_thm3}. 
\end{proof}

\begin{figure}[!t]
\centering
\subfloat[Ranging CRB under 16-QAM constellation.]{\includegraphics[width=\columnwidth]{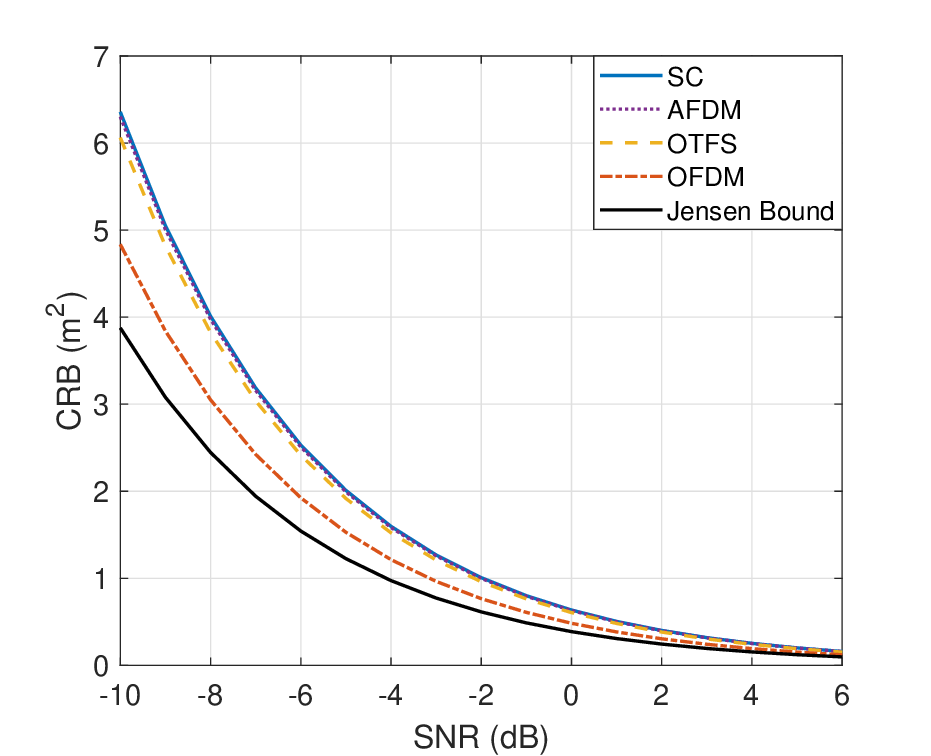}
\label{fig:crb_snr_16qam} }
\vspace{0.1in}
\subfloat[Ranging CRB under 16-PSK constellation.]{\includegraphics[width=\columnwidth]{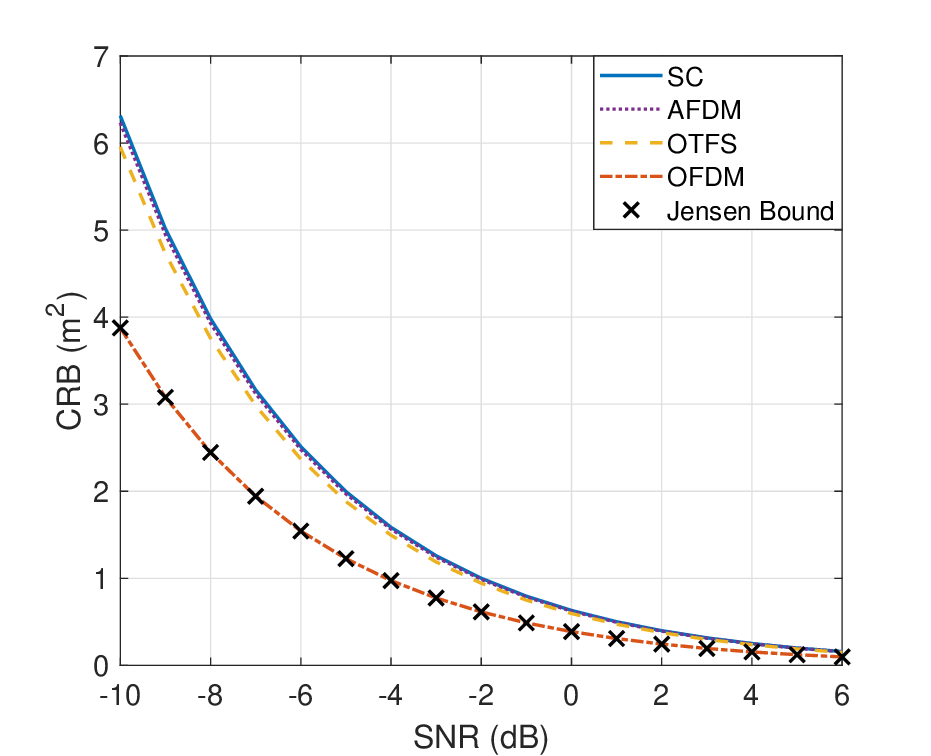}
\label{PCS_results_tradeoff}}
\caption{Ranging CRB for SC, AFDM, OTFS and OFDM under 16-QAM and 16-PSK constellations, $N = 128$, $L = 40$. The parameter settings for AFDM and OTFS are: $c_1 = \frac{1}{16}, c_2 = \frac{1}{8}$, and $N_1 = 32, N_2 = 4$.}
\label{fig:crb_snr}
\end{figure}

\subsection{Numerical Examples}
Fig.~\ref{fig:crb_snr} compares the ranging CRB of different orthogonal modulation schemes under 16-QAM and 16-PSK constellations. The numbers of transmitted symbols and targets are set as $N=128$ and $L=40$, respectively. For AFDM and OTFS, we set $c_1=1/16$, $c_2=1/8$, and $N_1=32$, $N_2=4$. The SNR is defined as $\mathrm{SNR}=1/\sigma^2$. In the simulation, we set the effective bandwidth as $B=160$ MHz. Under this bandwidth, one normalized delay unit corresponds to $\frac{c}{2B}=0.9375~\mathrm{m}$ in range, where $c$ is the speed of light, and the delay interval $\tau_\ell\in[0,127]$ corresponds to the range interval $[0,120~\mathrm{m}]$. Since the delay parameter $\tau_\ell$ in the considered model is measured in normalized sampling units, the normalized delay CRB is converted to physical ranging units by multiplying it by $\frac{c^2}{4B^2}$. Accordingly, the $40$ targets are randomly placed within this range interval.

It is observed from Fig.~\ref{fig:crb_snr} that OFDM consistently achieves the lowest ranging CRB among all tested waveforms under both constellations. Under 16-QAM, OFDM provides approximately $1$ dB SNR gain over the other representative waveforms for a given ranging accuracy. Under 16-PSK, OFDM exactly attains the Jensen bound, and the SNR gain increases to approximately $2$ dB. This exact achievability follows from the constant-modulus property of PSK symbols over the frequency domain under OFDM modulation, as discussed in Corollary \ref{CRB_OFDM_PSK}. In comparison, SC, AFDM, and OTFS show close CRB levels, and replacing 16-QAM with 16-PSK brings only marginal improvement for these waveforms. These observations confirm that the CRB advantage of OFDM originates from the waveform-dependent structure of the random FIM, and that the achievability of the Jensen bound requires the combination of unit-modulus constellations and the OFDM modulation basis.

\subsection{Discussions}
\subsubsection{Role of Bandwidth in ISAC Ranging}
The above results clarify that the advantage of OFDM does not come from a larger average waveform bandwidth, but from how random data symbols are placed in the frequency domain. This distinction is subtle because two notions of bandwidth are involved. Each OFDM basis function is localized on a single subcarrier and hence has zero individual RMS bandwidth, whereas the complete OFDM block still occupies the full system bandwidth. Indeed, for any unitary modulation matrix $\mathbf{U}$ and $\mathbf{x}=\mathbf{U}\mathbf{s}$, since $\mathbf{F}_N\mathbf{U}$ is unitary and $\mathbf{s}$ has i.i.d. unit-power entries, we have
\begin{align}
    \mathbb{E}\left\{\left|\mathbf{F}_N\mathbf{x}\right|^2\right\}
    =
    \mathbf{1}_N.
\end{align}
Thus, all unitary orthogonal modulations have the same average spectral power profile. On a uniform frequency grid, the RMS bandwidth induced by this average profile is
\begin{align}
    \overline{B}_{\rm rms}
    =
    \sqrt{\frac{1}{N}\sum_{n=0}^{N-1}\left(f_n-\frac{1}{N}\sum_k f_k\right)^2}
    =
    \sqrt{\frac{N^2-1}{12N^2}}.
\end{align}
Therefore, the comparison among OFDM, SC, OTFS, AFDM, and other waveforms is not a comparison between different average bandwidths. The essential difference is whether the full bandwidth is used in a \textit{stable way} for each data realization.

From the FIM perspective, all unitary bases induce the same average FIM $\bar{\mathbf{J}}$, and their difference appears only through the random fluctuation around this average. The Jensen gap between $\mathbb{E}\{\mathbf{J}^{-1}\}$ and $\bar{\mathbf{J}}^{-1}$ is therefore the price paid for symbol randomness. OFDM reduces this price by avoiding coherent mixing of independent data symbols across frequency bins. In contrast, in a frequency-spread basis, the signal on each frequency bin is a coherent sum of many random data symbols, leading to strong energy fluctuations across frequency. Destructive superposition may create spectral holes in some frequency bins, and these unfavorable realizations are heavily penalized because the CRB depends on the inverse of the random FIM. This explains why average RMS bandwidth alone is insufficient for characterizing the ranging CRB of data-bearing ISAC waveforms.

\subsubsection{Connection between CRB and Ambiguity Function}
It is also interesting to note that there is a deep connection between the ranging CRB and our previous analysis on the average ranging sidelobe of auto-correlation functions (ACFs) \cite{liu2024ofdm}. To see this, recall from~\cite[Eq.~(23)]{liu2024ofdm} that the average squared periodic ACF at lag $\ell$ can be written as
\begin{align}
    \mathbb{E}\{|\tilde r_\ell|^2\}
    = N^2\delta_{0,\ell}  + 
    N+(\kappa-2)\left\| \sum_{n}
    \mathbf{b}_n e^{-j\frac{2\pi \ell n}{N}}\right\|_2^2.
\end{align}
Consequently, since $\mathbf{B}$ is a permutation matrix under OFDM, which gives $\mathbf{b}_n^T\mathbf{b}_m=0$ for all $n\neq m$, the expected sidelobe level (ESL) gap at lag $\ell$ is
\begin{align}
    &\Delta{\rm ESL}_\ell(\mathbf{U})
    :=
    \mathbb{E}\{|\tilde r_\ell(\mathbf{U})|^2\}
    -
    \mathbb{E}\{|\tilde r_\ell({\rm OFDM})|^2\} \nonumber\\
    &=
    2(2-\kappa)
    \sum_{n<m}
    \left[
    1-\cos\left(\frac{2\pi\ell(m-n)}{N}\right)
    \right]\mathbf{b}_n^T\mathbf{b}_m .
\end{align}
Thus, the same pairwise overlap $\mathbf{b}_n^T\mathbf{b}_m$ that appears in the CRB analysis also controls the average sidelobe gap. Summing over all non-zero lags yields the expected integrated sidelobe level (EISL) gap as
\begin{align}
    \nonumber\Delta{\rm EISL}(\mathbf{U})
    &=
    2N(2-\kappa)
    \sum_{n<m}
    \mathbf{b}_n^T\mathbf{b}_m\\
    & =N(2-\kappa) \left(N -\left\|\mathbf{B}\right\|_F^2\right),
\end{align}
which is equivalent to the EISL expression in~\cite[Eq.~(27)]{liu2024ofdm}. In comparison, as indicated in \eqref{2nd_order_CRB_gap_closed_form}, the CRB gap is dominated by the second-order term in the form of
\begin{align}
    \nonumber\Delta{\rm CRB}(\mathbf{U}) &\approx \mathbb{E}\left\{\operatorname{Tr}\left[\mathbf{T}\mathbf{R}_2\left(\mathbf{U}\right)\right]\right\} - \mathbb{E}\left\{\operatorname{Tr}\left[\mathbf{T}\mathbf{R}_2\left(\mathbf{F}_N^H\right)\right]\right\}\\
    & =  (2-\kappa)
    \sum_{n<m}
    \mu_{nm}\mathbf{b}_n^T\mathbf{b}_m ,
    \end{align}
where $\mu_{nm}$, defined in Appendix \ref{proof_theorem_1}, is purely determined by the target geometry. Therefore, the EISL gap is an unweighted accumulation of the pairwise spectral overlaps, whereas the CRB gap is a geometry-weighted accumulation of the same quantities. This shows that average sidelobes and expected CRB are two different metrics of the same symbol-induced spectral fluctuations. The former measures their correlation-domain energy, while the latter measures their impact on the random FIM after being weighted by the target geometry and amplified through matrix inversion.

\subsubsection{Near-OFDM Waveforms}
The local optimality result under large $N$ gives the same message from a complementary perspective. The comparison with $\alpha$-spread waveforms covers bases that introduce a non-vanishing amount of symbol mixing across the frequency domain. By contrast, the local result focuses on the near-OFDM regime, where such mixing may vanish with $N$ and is therefore not captured by the frequency-spread condition. It shows that, even in this regime, a small departure from OFDM cannot systematically reduce the CRB in the large-$N$ limit. In other words, the apparent freedom to slightly mix neighboring subcarriers does not translate into a lower ranging CRB, because such mixing increases the leading fluctuation penalty induced by random symbols.

This complementarity is important for interpreting the scope of the theory. The frequency-spread result rules out broad classes of non-OFDM bases such as SC, OTFS, AFDM, and other spectrum-spread constructions, while the local result rules out small, carefully designed perturbations around OFDM. The two optimality results leave only a narrow unresolved region in the unitary waveform space. Any improvement in the ranging CRB, if it exists, would have to come from a non-localized yet non-spread waveform basis, for example a highly structured mixing of nearby subcarriers whose RMS bandwidth still vanishes with $N$. Such waveforms are neither conventional spectrum-spread modulations nor small perturbations of OFDM in the Riemannian sense. This observation provides a concrete interpretation of the remaining gap between the present results and a full global optimality theorem.

The common lesson learned from these results is that, within linear unitary modulations, changing the modulation basis alone from OFDM offers limited room for performance improvement. More substantial gains are more likely to come from pulse shaping, constellation shaping through $\kappa$ and higher moments~\cite{liu2026jsac_tut}, or waveform designs beyond the linear orthogonal framework.
\begin{conjecture}[Global Optimality of OFDM over the Almost-Sure Set]
    Under the same assumptions as Theorems \ref{asm_opt_thm}--\ref{hessian_thm}, OFDM is the global minimizer of the ranging CRB over $\Omega_N$, i.e., for all $\mathbf{U}\in\mathbb{U}(N)$,
    \begin{equation}
        \mathbb{E}\!\left[f(\mathbf{U})\left|\Omega_N\right.\right]\ge \mathbb{E}\!\left[f(\mathbf{F}_N^{H})\left|\Omega_N\right.\right].
    \end{equation}
\end{conjecture}

\section{Conclusions}\label{conclusion_sec}
This paper investigated how the orthogonal modulation basis affects the ranging CRB in ISAC systems with random communication symbols. We derived a Jensen-type lower bound that is independent of both the modulation basis and the constellation, and showed that CP-OFDM exactly attains this bound under PSK constellations. For QAM and more general sub-Gaussian constellations, we proved that CP-OFDM achieves a lower ranging CRB than any frequency-spread orthogonal waveform in the large-sample regime, with the gap governed by the dominant second-order term in the inverse-FIM expansion. The same superiority was also extended to amplitude estimation and to the full CRB for joint delay and amplitude estimation. We further showed that CP-OFDM is a stationary point for finite $N$, and that its Riemannian Hessian is positive semidefinite for sufficiently large $N$. This establishes the asymptotic local optimality of CP-OFDM. Numerical results confirmed the theoretical findings, with OFDM yielding the lowest CRB among representative waveforms, including SC, OTFS, and AFDM.

\appendices
\section{Proof of Lemma \ref{EFIM_order}}\label{proof_lemma_2}
Let us first partition $\bar{\mathbf{J}}$ following the convention in \eqref{J_partition}, such that the $\left(k,\ell\right)$-th entry of each block can be given as
\begin{align}
    \left[\bar{\mathbf{J}}_{\tau\tau}\right]_{k,\ell} &\nonumber= \frac{{2}}{{\sigma ^2}}\operatorname{Re} \left\{\beta_k^\ast\beta_{\ell} \dot{\mathbf{a}}_k^H\dot{\mathbf{a}}_\ell \right\}\\
    &\nonumber = \frac{{8\pi^2}}{{\sigma ^2}}\operatorname{Re} \left\{\beta_k^\ast\beta_{\ell} \sum_{n=0}^{N-1}\frac{n^2}{N^2}e^{-\frac{j2\pi n\left(\tau_\ell-\tau_k\right)}{N}} \right\},\\
    \left[\bar{\mathbf{J}}_{\tau a}\right]_{k,\ell} &\nonumber= \frac{{2}}{{\sigma ^2}}\operatorname{Re} \left\{\beta_k^\ast \dot{\mathbf{a}}_k^H{\mathbf{a}}_\ell \right\}\\
    &\nonumber = \frac{{4\pi}}{{\sigma ^2}}\operatorname{Re} \left\{j\beta_k^\ast \sum_{n=0}^{N-1}\frac{n}{N}e^{-\frac{j2\pi n\left(\tau_\ell-\tau_k\right)}{N}} \right\},\\
    \left[\bar{\mathbf{J}}_{\tau b}\right]_{k,\ell} &\nonumber= \frac{{2}}{{\sigma ^2}}\operatorname{Re} \left\{j\beta_k^\ast \dot{\mathbf{a}}_k^H{\mathbf{a}}_\ell \right\}\\
    &\nonumber = \frac{{4\pi}}{{\sigma ^2}}\operatorname{Re} \left\{-\beta_k^\ast \sum_{n=0}^{N-1}\frac{n}{N}e^{-\frac{j2\pi n\left(\tau_\ell-\tau_k\right)}{N}} \right\},\\
    \left[\bar{\mathbf{J}}_{aa}\right]_{k,\ell} &\nonumber= \left[\bar{\mathbf{J}}_{bb}\right]_{k,\ell}=\frac{{2}}{{\sigma ^2}}\operatorname{Re} \left\{{\mathbf{a}}_k^H{\mathbf{a}}_\ell \right\} \\&\nonumber= \frac{{2}}{{\sigma ^2}}\operatorname{Re} \left\{ \sum_{n=0}^{N-1}e^{-\frac{j2\pi n\left(\tau_\ell-\tau_k\right)}{N}} \right\},\\
    \left[\bar{\mathbf{J}}_{ab}\right]_{k,\ell} &= \frac{{2}}{{\sigma ^2}}\operatorname{Re} \left\{j{\mathbf{a}}_k^H{\mathbf{a}}_\ell \right\} = \frac{{2}}{{\sigma ^2}}\operatorname{Re} \left\{ j\sum_{n=0}^{N-1}e^{-\frac{j2\pi n\left(\tau_\ell-\tau_k\right)}{N}} \right\},
\end{align}
Define
\begin{align}
    &\nonumber I_{\tau\tau}\left(\tau\right):=\lim_{N\to\infty}\frac{1}{N}\sum_{n=0}^{N-1} \left(\frac{n}{N}\right)^2e^{\frac{-j2\pi n\tau}{N}} = \int_0^1 {x^2}e^{-j2\pi x \tau} dx,\\
    &\nonumber I_{\beta\beta}\left(\tau\right):=\lim_{N\to\infty}\frac{1}{N}\sum_{n=0}^{N-1} e^{\frac{-j2\pi n\tau}{N}} = \int_0^1e^{-j2\pi x \tau} dx,\\
    &I_{\tau \beta}\left(\tau\right):=\lim_{N\to\infty}\frac{1}{N}\sum_{n=0}^{N-1} \frac{n}{N}e^{\frac{-j2\pi n\tau}{N}} = \int_0^1xe^{-j2\pi x \tau} dx.
\end{align}
Then
\begin{align}\label{infty_entries}
    &\nonumber\lim_{N\to\infty}\frac{1}{N}\left[\bar{\mathbf{J}}_{\tau\tau}\right]_{k,\ell}= \frac{{8{\pi ^2}}}{{\sigma ^2}}\operatorname{Re} \left\{\beta_k^\ast\beta_{\ell} I_{\tau\tau}\left(\tau_\ell - \tau_k\right) \right\},\\
    &\nonumber\lim_{N\to\infty}\frac{1}{N}\left[\bar{\mathbf{J}}_{\tau a}\right]_{k,\ell}= \frac{{4{\pi}}}{{\sigma ^2}}\operatorname{Re} \left\{j\beta_k^\ast I_{\tau \beta}\left(\tau_\ell - \tau_k\right) \right\},\\
    &\nonumber\lim_{N\to\infty}\frac{1}{N}\left[\bar{\mathbf{J}}_{\tau b}\right]_{k,\ell}= \frac{{4{\pi}}}{{\sigma ^2}}\operatorname{Re} \left\{-\beta_k^\ast I_{\tau \beta}\left(\tau_\ell - \tau_k\right) \right\},\\
    &\nonumber\lim_{N\to\infty}\frac{1}{N}\left[\bar{\mathbf{J}}_{aa}\right]_{k,\ell}= \frac{{2}}{{\sigma ^2}}\operatorname{Re} \left\{I_{\beta \beta}\left(\tau_\ell - \tau_k\right) \right\},\\
    &\lim_{N\to\infty}\frac{1}{N}\left[\bar{\mathbf{J}}_{ab}\right]_{k,\ell}= \frac{{2}}{{\sigma ^2}}\operatorname{Re} \left\{jI_{\beta \beta}\left(\tau_\ell - \tau_k\right) \right\}.
\end{align}
We further define the limit matrix
\begin{equation}\label{limit_barJ}
    {\bar {\mathbf{J}}}_{\infty}:=\lim_{N\to\infty}\frac{1}{N}{\bar {\mathbf{J}}},
\end{equation}
where the entries of ${\bar {\mathbf{J}}}_{\infty}$ are given in \eqref{infty_entries}. 
Observe that for any $\mathbf{d}\in\mathbb{R}^{3L}$, we have
\begin{align}\label{semidefiniteness_J_bar} 
    &\nonumber\mathbf{d}^T{\bar {\mathbf{J}}}_{\infty}\mathbf{d} = \\&\frac{2}{\sigma^2}\int_0^1\left|\sum_{\ell}\left(d_{L+\ell}+jd_{2L+\ell}-j2\pi \beta_\ell d_\ell x\right)e^{-j2\pi\tau_\ell x}\right|^2dx \ge 0.
\end{align}
Let us assume the equality in \eqref{semidefiniteness_J_bar} is attained at some non-zero $\mathbf{d}$. Since
\begin{equation}
    h(x):=\left|\sum_{\ell}\left(d_{L+\ell}+jd_{2L+\ell}-j2\pi \beta_\ell d_\ell x\right)e^{-j2\pi\tau_\ell x}\right|^2
\end{equation}
is continuous and non-negative on $(0,1]$, the equality holds only if $h(x)\equiv 0$ on $(0,1]$, namely,
\begin{equation}
    \sum_{\ell}\left(d_{L+\ell}+jd_{2L+\ell}-j2\pi \beta_\ell d_\ell x\right)e^{-j2\pi\tau_\ell x} \equiv 0, \;\;x\in(0,1].
\end{equation}
By noting that the delays $\{\tau_\ell\}_{\ell=1}^L$ are distinct, the functions $\{e^{-j2\pi \tau_\ell x}\}_{\ell=1}^L$ are linearly independent on $(0,1]$. This implies 
\begin{equation}\label{constant_zero_condition}
    d_{L+\ell}+jd_{2L+\ell}-j2\pi \beta_\ell d_\ell x \equiv 0, \;\;x\in(0,1],\quad \forall \ell.
\end{equation}
Since $\beta_\ell \ne 0$ for all $\ell$, \eqref{constant_zero_condition} holds for all $x\in(0,1]$ if and only if $\mathbf{d} = \mathbf{0}$, contradicting the assumption. Therefore,
\begin{equation}
    \mathbf{d}^T\bar{\mathbf{J}}_{\infty}\mathbf{d}>0
\end{equation}
for any non-zero real vector $\mathbf{d}$, and hence $\bar{\mathbf{J}}_{\infty}$ is positive definite. In particular, $\lambda_{\min}\!\left(\bar{\mathbf{J}}_{\infty}\right)>0$. By continuity of the eigenvalue functions, we further have
\begin{align}
    &\mathop {\lim }\limits_{N \to \infty } \lambda_{\min}\left(\frac{1}{N}{\bar {\mathbf{J}}}\right) = \lambda_{\min}\left(\bar{\mathbf{J}}_{\infty}\right) >0,\\
    &\mathop {\lim }\limits_{N \to \infty } \lambda_{\max}\left(\frac{1}{N}{\bar {\mathbf{J}}}\right) = \lambda_{\max}\left(\bar{\mathbf{J}}_{\infty}\right) >0.
\end{align}
Let
\begin{equation}
    c_0 = \frac{1}{2}\lambda_{\min}\left({\bar {\mathbf{J}} _{\infty }}\right),\quad c_1 = 2\lambda_{\max}\left({\bar {\mathbf{J}} _{\infty }}\right).
\end{equation}
Then for sufficiently large $N$,
\begin{equation}
    \lambda_{\min}\left({\bar {\mathbf{J}}}\right) \ge c_0 N,\quad \lambda_{\max}\left({\bar {\mathbf{J}}}\right) \le c_1 N,
\end{equation}
which completes the proof.

\section{Proof of Lemma \ref{Z_HW_concerntration}}\label{proof_lemma_3}
Since $s$ is sub-Gaussian, there exists a constant $R$, such that $\|s\|_{\psi_2}\le R$, where $\|\cdot\|_{\psi_2}$ denotes the sub-Gaussian norm, defined for real random variable $x$ as \cite{rudelson2013hanson}
\begin{equation}
    \|x\|_{\psi_2}:= \sup_{p\ge 1}p^{-1/2}\left[\mathbb{E}\left(\left|x\right|^p\right)\right]^{1/p}
\end{equation}
and for complex random variable as
\begin{equation}
    \left\|s\right\|_{\psi_2} = \max\left\{\left\|\operatorname{Re}(s)\right\|_{\psi_2},\left\|\operatorname{Im}(s)\right\|_{\psi_2}\right\}.
\end{equation}
Observe that the perturbation matrix $\bm{\Delta}$ can be expressed as
\begin{equation}
    \bm{\Delta} = \mathbf{J} - \bar{\mathbf{J}} = \frac{2}{\sigma^2} \operatorname{Re} \left\{ \mathbf{H}^H \operatorname{Diag}\left( |\mathbf{Q}^H\mathbf{s}|^2 - \mathbf{1}_N \right) \mathbf{H} \right\}.
\end{equation}
Accordingly, the $(k, \ell)$-th entry of $\bm{\Delta}$ is given by
\begin{equation}
    \left[\bm{\Delta}\right]_{k,\ell} = \operatorname{Re} \left\{ \mathbf{s}^H \mathbf{G}_{k\ell} \mathbf{s} - \mathbb{E}\left( \mathbf{s}^H \mathbf{G}_{k\ell} \mathbf{s} \right) \right\},
\end{equation}
where $\mathbf{G}_{k\ell}$ is defined as
\begin{equation}
    \mathbf{G}_{k\ell} = \frac{2}{\sigma^2} \mathbf{Q} \mathbf{M}_{k\ell} \mathbf{Q}^H,
\end{equation}
and $\mathbf{M}_{k\ell}:=\operatorname{Diag}({\mathbf{h}}_{k}^\ast\odot{\mathbf{h}}_{\ell})$. Here, ${\mathbf{h}}_{k} \in \mathbb{C}^N$ represents the $k$-th column of the Jacobian matrix $\mathbf{H}$, which is given by
\begin{equation}
    \mathbf{h}_k=
    \begin{cases}
        \beta_k \dot{\mathbf{a}}_k, & 1\le k\le L,\\[1mm]
        \mathbf{a}_{k-L}, & L+1\le k\le 2L,\\[1mm]
        j\,\mathbf{a}_{k-2L}, & 2L+1\le k\le 3L.
    \end{cases}
\end{equation}
By the Hanson-Wright inequality \cite{rudelson2013hanson}, for quadratic forms of sub-Gaussian random vector $\mathbf{s}$, there exists a positive constant $c>0$, such that
    \begin{align}
        &\nonumber \Pr\left\{\textstyle|\left[\bm{\Delta}\right]_{k,\ell}|\ge t\right\} \\&\nonumber= \Pr\left\{\textstyle\left|\operatorname{Re}\left\{\mathbf{s}^H\mathbf{G}_{k\ell}\mathbf{s}-\mathbb{E}\left(\mathbf{s}^H\mathbf{G}_{k\ell}\mathbf{s}\right)\right\}\right|\ge t\right\}\\
        &\nonumber\le 2\exp\left\{-c\min\left(\frac{t^2}{R^4\left\|\mathbf{G}_{k\ell}\right\|_F^2},\frac{t}{R^2\left\|\mathbf{G}_{k\ell}\right\|_2}\right)\right\}.
    \end{align}
It is straightforward to verify that the Frobenius norm of $\mathbf{G}_{k\ell}$ satisfies
    \begin{align}
        &\|\mathbf{G}_{k\ell}\|_F^2 = \frac{4}{\sigma^4}\|\mathbf{M}_{k\ell}\|_F^2 =\frac{4}{\sigma^4}\|{\mathbf{h}}_{k}^\ast\odot{\mathbf{h}}_{\ell}\|^2 \nonumber \\&=
        \begin{cases}\nonumber
             \frac{64\pi^4|\beta_k|^2|\beta_\ell|^2}{\sigma^4}\sum_{n=1}^N \frac{(n-1)^4}{N^4}, & \scriptstyle 1\le k\le L,\;\;1\le \ell\le L\\[1mm]
             \frac{16\pi^2|\beta_k|^2}{\sigma^4}\sum_{n=1}^N \frac{(n-1)^2}{N^2}, & \scriptstyle1\le k\le L,\;\; L+1\le \ell\le 3L,\\[1mm]
             \frac{16\pi^2|\beta_\ell|^2}{\sigma^4}\sum_{n=1}^N \frac{(n-1)^2}{N^2}, & \scriptstyle L+1\le k\le 3L,\;\; 1\le \ell\le L,\\[1mm]
            \frac{4N}{\sigma^4}, &\scriptstyle L+1\le k\le 3L,\;\;L+1\le \ell\le 3L
        \end{cases}\\
        & = \Theta(N),
    \end{align}
while its spectral norm may be calculated as
    \begin{align}
        \|\mathbf{G}_{k\ell}\|_2 &=  \frac{2}{\sigma^2}\|\mathbf{M}_{k\ell}\|_2 \nonumber \\
        &=
        \begin{cases}\nonumber
            \frac{8\pi^2|\beta_k||\beta_\ell|(N-1)^2}{\sigma^2N^2}, & \scriptstyle1\le k\le L,\;\;1\le \ell\le L,\\[1mm]
            \frac{4\pi|\beta_k|(N-1)}{\sigma^2N}, & \scriptstyle1\le k\le L,\;\; L+1\le \ell\le 3L,\\[1mm]
             \frac{4\pi|\beta_\ell|(N-1)}{\sigma^2N}, & \scriptstyle L+1\le k\le 3L,\;\; 1\le \ell\le L,\\[1mm]
            \frac{2}{\sigma^2}, &\scriptstyle L+1\le k\le 3L,\;\;L+1\le \ell\le 3L,
        \end{cases}\\
        & = \Theta(1).
    \end{align}
Let $\Delta_{\max} = {\max }_{k,\ell} |\left[\bm{\Delta}\right]_{k,\ell}|$ denote the entry of $\bm{\Delta}$ with the maximum magnitude. Noting that
    \begin{equation}
        \|\bm{\Delta}\|_F^2 = \sum_{k,\ell} |\left[\bm{\Delta}\right]_{k,\ell}|^2 \le 9L^2 |\Delta_{\max}|^2,
    \end{equation}
we may bound the tail probability of the Frobenius norm using the union bound:
    \begin{align}
        &\Pr\left\{ \|\bm{\Delta}\|_F \ge t \right\} \le \Pr\left\{ |\Delta_{\max}| \ge \frac{t}{3L} \right\} \nonumber \\
        &\le \sum_{k,\ell} \Pr\left\{ |\left[\bm{\Delta}\right]_{k,\ell}| \ge \frac{t}{3L} \right\} \le 9L^2 \Pr\left\{ |\left[\bm{\Delta}\right]_{k^{\prime}\ell^{\prime}}| \ge \frac{t}{3L} \right\},
    \end{align}
    where 
    \begin{equation}
        \left(k^{\prime}, \ell^{\prime}\right) =\arg \mathop {\max }\limits_{k,\ell}\Pr\left\{ |\left[\bm{\Delta}\right]_{k,\ell}|^2 \ge \frac{t^2}{9L^2} \right\}.
    \end{equation}
By absorbing all constants into $c$ and setting $t = \varepsilon^{\prime}N$ for some $\varepsilon^{\prime} > 0$, we arrive at
    \begin{equation}
        \Pr\left\{ \|\bm{\Delta}\|_F \ge \varepsilon^{\prime}N \right\} \le 18L^2 \exp\left\{ -cN \min({\varepsilon^{\prime}}^2, \varepsilon^{\prime}) \right\}.
    \end{equation}
For sufficiently large $N$, recall that $\|\bar{\mathbf{J}}^{-1}\|_2 \le \frac{1}{c_0N}$. Consequently, the spectral norm of $\mathbf{Z}$ satisfies
    \begin{equation}
        \|\mathbf{Z}\|_2 \le \|\mathbf{Z}\|_F \le \|\bar{\mathbf{J}}^{-1}\|_2 \|\bm{\Delta}\|_F \le \frac{1}{c_0N} \|\bm{\Delta}\|_F.
    \end{equation}
By defining $\varepsilon = \frac{\varepsilon^\prime}{c_0}$ and consolidating the constants into $c$, the concentration inequality \eqref{concentration_sub_Gaussian} is immediately established.

Next, we evaluate the moments of $\|\mathbf{Z}\|_2$. The $k$-th moment can be expressed in terms of the tail probability as~\cite{durrett2019probability}
    \begin{equation}
        \mathbb{E}\left( \|\mathbf{Z}\|_2^k \right) = k \int_0^\infty t^{k-1} \Pr\left( \|\mathbf{Z}\|_2 \ge t \right) dt.
    \end{equation}
For the interval $0 < t < 1$, we have $\min(t^2, t) = t^2$. It follows that
    \begin{equation}
        \int_0^1 t^{k-1} e^{-cNt^2} dt \le \int_0^\infty t^{k-1} e^{-cNt^2} dt = \frac{1}{2} (cN)^{-\frac{k}{2}} \Gamma(k/2),
    \end{equation}
where $\Gamma(\cdot)$ denotes the Gamma function. Conversely, for $t \ge 1$, we have $\min(t^2, t) = t$, yielding
    \begin{equation}
        \int_1^\infty t^{k-1} e^{-cNt} dt \le \int_0^\infty t^{k-1} e^{-cNt} dt = (cN)^{-k} \Gamma(k).
    \end{equation}
By combining these bounds and incorporating the constants, we obtain
    \begin{align}
        \mathbb{E}\left( \|\mathbf{Z}\|_2^k \right)&\le 18L^2 k \left[ (cN)^{-k} \Gamma(k) + \frac{1}{2} (cN)^{-k/2} \Gamma(k/2) \right] \nonumber \\
        &= O(N^{-k/2}).
    \end{align}
This completes the proof.

\section{Proof of Lemma \ref{bistochastic_prop}}\label{proof_prop_bistochastic}
Let us first re-express the FIM as
        \begin{align}\label{FIM_decomposition}
            \mathbf{J} = \sum_{n=1}^N \tilde{s}_n\mathbf{C}_n,
        \end{align}
        where 
        \begin{align}\label{Cn_def}
           \tilde{s}_n := \left|\mathbf{q}_n^H \mathbf{s}\right|^2,\quad\mathbf{C}_n:=\frac{2}{\sigma^2}
        \operatorname{Re}\left(\widetilde{\mathbf{h}}_{n}\widetilde{\mathbf{h}}_{n}^H\right), 
        \end{align}
        in which $\mathbf{q}_n$ denotes the $n$th column of $\mathbf{Q}$, and $\widetilde{\mathbf{h}}_{n}\in\mathbb{C}^{3L}$ is the $n$th column of $\mathbf{H}^H$.
        Define 
        \begin{equation}\label{tilde_Jacobian}
            \widetilde{\mathbf{H}}: = \left[\operatorname{vec}\left(\mathbf{C}_1\right),\operatorname{vec}\left(\mathbf{C}_2\right),\ldots,\operatorname{vec}\left(\mathbf{C}_N\right)\right].
        \end{equation}
        It follows that
        \begin{equation}
            \mathbf{v} = \sum_{n=1}^N \tilde{s}_n\operatorname{vec}\left(\mathbf{C}_n\right) = \widetilde{\mathbf{H}}\tilde{\mathbf{s}}.
        \end{equation}
        Consequently,
        \begin{equation}\label{covaraince_Jacobian}
            \bm{\Sigma}_{\mathbf{v}} = \widetilde{\mathbf{H}}\bm{\Sigma}_{\tilde{\mathbf{s}}}\widetilde{\mathbf{H}}^T,
        \end{equation}
        where $\bm{\Sigma}_{\tilde{\mathbf{s}}} = \operatorname{Cov}\left(\tilde{\mathbf{s}},\tilde{\mathbf{s}}\right)$, with $\tilde{\mathbf{s}} = \left[\tilde{s}_1,\tilde{s}_2,\ldots,\tilde{s}_N\right]^T$. 
        Therefore, it suffices to show
        \begin{equation}\label{inequality_bi_stochastic}
            \bm{\Sigma}_{\tilde{\mathbf{s}}}\left(\mathbf{U}\right)\succeq\bm{\Sigma}_{\tilde{\mathbf{s}}}\left(\mathbf{F}_N^H\right).
        \end{equation}
        Note that
        \begin{align}\label{cov_n_m}
            &\nonumber\mathbb{E}\left(\left|\mathbf{q}_n^H \mathbf{s}\right|^2\left|\mathbf{q}_m^H \mathbf{s}\right|^2\right) = \mathbb{E}\left(\sum_{k,p} q_{nk}^\ast s_k s_p^\ast q_{np}\sum_{i\ell} q_{mi}^\ast s_i s_\ell^\ast q_{m\ell}\right)\\
            & = \sum_{i,\ell,k,p}q_{mi}^\ast q_{m\ell}q_{nk}^\ast q_{np}\mathbb{E}\left( s_i s_\ell^\ast s_k s_p^\ast\right),
        \end{align}
        where
        \begin{equation}\label{4s}
            \mathbb{E}\left( s_i s_\ell^\ast s_k s_p^\ast\right) = \delta_{i,\ell}\delta_{k,p}+\delta_{i,p}\delta_{\ell,k}+\left(\kappa-2\right)\delta_{i,\ell}\delta_{\ell,k}\delta_{k,p}.
        \end{equation}
        Substituting \eqref{4s} into \eqref{cov_n_m} yields
        \begin{align}\label{cov_n_m_1}
            &\nonumber\mathbb{E}\left(\left|\mathbf{q}_n^H \mathbf{s}\right|^2\left|\mathbf{q}_m^H \mathbf{s}\right|^2\right)  \\
            &\nonumber=\left\|\mathbf{q}_n\right\|^2\left\|\mathbf{q}_m\right\|^2+\left|\mathbf{q}_n^H\mathbf{q}_m\right|^2+\left(\kappa-2\right){\left|\mathbf{q}_n\right|^2}^T\left|\mathbf{q}_m\right|^2\\
            & = 1+\delta_{n,m}+\left(\kappa-2\right){\left|\mathbf{q}_n\right|^2}^T\left|\mathbf{q}_m\right|^2.
        \end{align}
        Since
        \begin{align}
            \nonumber\left[\bm{\Sigma}_{\tilde{\mathbf{s}}}\right]_{n,m} &= \mathbb{E}\left(\left|\mathbf{q}_n^H \mathbf{s}\right|^2\left|\mathbf{q}_m^H \mathbf{s}\right|^2\right) - \mathbb{E}\left(\left|\mathbf{q}_n^H \mathbf{s}\right|^2\right)\mathbb{E}\left(\left|\mathbf{q}_m^H \mathbf{s}\right|^2\right)\\
            & = \delta_{n,m}+\left(\kappa-2\right){\left|\mathbf{q}_n\right|^2}^T\left|\mathbf{q}_m\right|^2,
        \end{align}
        we have
        \begin{equation}\label{sigma_s_expression}
            \bm{\Sigma}_{\tilde{\mathbf{s}}} = \mathbf{I}_N+\left(\kappa-2\right){\mathbf{B}}^T\mathbf{B}.
        \end{equation}
        Since $\mathbf{B} := \left|\mathbf{Q}\right|^2$ is attained through entrywise square of a unitary matrix, it is a doubly stochastic matrix, satisfying
        \begin{equation}
            \left\|\mathbf{B}\right\|_2\le 1.
        \end{equation}
        Hence
        \begin{equation}
            {\mathbf{B}}^T\mathbf{B}\preceq \mathbf{I}_N,
        \end{equation}
        where the equality holds if and only if $\mathbf{B}$ is a permutation matrix. This is only possible if
        \begin{equation}
            \mathbf{U} = \mathbf{F}_N^H\bm{\Pi},
        \end{equation}
        where $\bm{\Pi}\in\mathbb{C}^{N\times N}$ is a complex permutation matrix, corresponding to the OFDM waveform up to subcarrier permutation and phase rotation. Since $\kappa < 2$, the inequality \eqref{inequality_bi_stochastic} follows immediately, completing the proof.

\section{Proof of Theorem \ref{asm_opt_thm}}\label{proof_theorem_1}
    By noting the resolvent-type expansion in \eqref{J_expansion}, we have
    \begin{align}\label{trace_expansion}
         &\mathbb{E}\left[f\left(\mathbf{U}\right)\left|\Omega_N\right.\right] \nonumber= \operatorname{Tr}\left(\mathbf{T}{{{\bar {\mathbf{J}} }^{-1}}}\right) -\mathbb{E}\left\{\operatorname{Tr}\left(\mathbf{T}\mathbf{R}_1\right)\left|\Omega_N\right.\right\}\\&\qquad\quad+ \mathbb{E}\left\{\operatorname{Tr}\left(\mathbf{T}\mathbf{R}_2\right)\left|\Omega_N\right.\right\}-\mathbb{E}\left\{\operatorname{Tr}\left(\mathbf{T}\mathbf{R}_3\right)\left|\Omega_N\right.\right\}.
    \end{align}
    Observe that the zero-order term $\operatorname{Tr}\left(\mathbf{T}{{{\bar {\mathbf{J}} }^{-1}}}\right)$ is exactly the Jensen bound, and is thus irrelevant to the adopted modulation waveform. In what follows, we present an analysis of the remaining terms.
    \subsection{First-Order Term}\label{1st_term}
    Since $\mathbb{E}\left(\bm{\Delta}\right) = \mathbf{0}$,
    \begin{align}
        \mathbb{E}\left(\mathbf{T}\mathbf{R}_1\right) &\nonumber= \Pr\left\{\Omega_N\right\}\mathbb{E}\left(\mathbf{T}\mathbf{R}_1\left|\Omega_N\right.\right) \\
        &\quad\;\;+ \Pr\left\{\Omega_N^c\right\}\mathbb{E}\left(\mathbf{T}\mathbf{R}_1\left|\Omega_N^c\right.\right)= \mathbf{0},
    \end{align}
    yielding
    \begin{align}
        \nonumber\mathbb{E}\left\{\operatorname{Tr}\left(\mathbf{T}\mathbf{R}_1\right)\left|\Omega_N\right.\right\} &= -\frac{\Pr\left\{\Omega_N^c\right\}\mathbb{E}\left\{\operatorname{Tr}\left(\mathbf{T}\mathbf{R}_1\right)\left|\Omega_N^c\right.\right\}}{\Pr\left\{\Omega_N\right\}}\\
        &= -\frac{\mathbb{E}\left\{\operatorname{Tr}\left(\mathbf{T}\mathbf{R}_1\right)\cdot 1_{\Omega_N^c}\right\}}{\Pr\left\{\Omega_N\right\}},
    \end{align}
    where $1_{\Omega_N^c}$ is an indicator function that takes value $1$ if the event $\Omega_N^c$ occurs, and $0$ otherwise. Consequently, by the Cauchy-Schwarz inequality,
    \begin{align}
        &\nonumber\left|\mathbb{E}\left\{\operatorname{Tr}\left(\mathbf{T}\mathbf{R}_1\right)\cdot 1_{\Omega_N^c}\right\}\right|\le \sqrt{\mathbb{E}\left(\left|\operatorname{Tr}\left(\mathbf{T}\mathbf{R}_1\right)\right|^2\right)\mathbb{E}\left(1_{\Omega_N^c}^2\right)}\\
        &\nonumber =\sqrt{\mathbb{E}\left(\left|\operatorname{Tr}\left(\mathbf{T}\mathbf{R}_1\right)\right|^2\right)\Pr\left\{\Omega_N^c\right\}}\\
        &\le L\sqrt{\mathbb{E}\left(\left\|\mathbf{T}\mathbf{R}_1\right\|_2^2\right)\Pr\left\{\Omega_N^c\right\}},
    \end{align}
    and thus
    \begin{align}
        \left|\mathbb{E}\left\{\operatorname{Tr}\left(\mathbf{T}\mathbf{R}_1\right)\left|\Omega_N\right.\right\}\right|\le \frac{L}{\Pr\left\{\Omega_N\right\}}\sqrt{\mathbb{E}\left(\left\|\mathbf{T}\mathbf{R}_1\right\|_2^2\right)\Pr\left\{\Omega_N^c\right\}}.
    \end{align}
    According to Lemma 3, since
    \begin{align}
        \nonumber\mathbb{E}\left(\left\|\mathbf{T}\mathbf{R}_1\right\|_2^2\right) &= \mathbb{E}\left(\left\|\mathbf{T}\mathbf{Z}{{{\bar {\mathbf{J}} }^{-1}}}\right\|_2^2\right)\le  \left\|\mathbf{T}\right\|_2^2\left\|{{{\bar {\mathbf{J}} }^{-1}}}\right\|_2^2\mathbb{E}\left(\left\|\mathbf{Z}\right\|_2^2\right) \\
        & =\Theta\left(N^{-2}\right)O\left(N^{-1}\right)= O\left(N^{-3}\right),
    \end{align}
    we have
    \begin{align}
        \nonumber\left|\mathbb{E}\left\{\operatorname{Tr}\left(\mathbf{T}\mathbf{R}_1\right)\left|\Omega_N\right.\right\}\right|&\le\frac{L}{\Pr\left\{\Omega_N\right\}}\sqrt{O\left(N^{-3}\right)\Pr\left\{\Omega_N^c\right\}}\\
        &= O\left(N^{-3/2}e^{-cN}\right).
    \end{align}

    \subsection{Second-Order Term}
    By a similar argument, it is straightforward to show that
    \begin{equation}\label{R2_order_omega_n}
        \mathbb{E}\left\{\operatorname{Tr}\left(\mathbf{T}\mathbf{R}_2\right)\left|\Omega_N\right.\right\} =\frac{\mathbb{E}\left\{\operatorname{Tr}\left(\mathbf{T}\mathbf{R}_2\right)\right\}}{\Pr\left\{\Omega_N\right\}}+O\left(N^{-2}e^{-cN}\right),
    \end{equation}
    which suggests that the order of $\mathbb{E}\left\{\operatorname{Tr}\left(\mathbf{T}\mathbf{R}_2\right)\left|\Omega_N\right.\right\}$ is dominated by $\mathbb{E}\left\{\operatorname{Tr}\left(\mathbf{T}\mathbf{R}_2\right)\right\}$.
    Note that
    \begin{align}\label{2nd_term_analysis}
        &\nonumber\mathbb{E}\left\{\operatorname{Tr}\left(\mathbf{T}\mathbf{R}_2\right)\right\} = \mathbb{E}\left\{\operatorname{Tr}\left(\mathbf{T}{{{\bar{\mathbf{J}} }^{-1}}}\bm{\Delta}{{{\bar {\mathbf{J}} }^{-1}}}\bm{\Delta}{{{\bar {\mathbf{J}} }^{-1}}}\right)\right\} \\
        &\nonumber= \mathbb{E}\left\{\operatorname{vec}\left(\bm{\Delta}\right)^T\mathbf{W}\operatorname{vec}\left(\bm{\Delta}\right)\right\} \\
        &= \operatorname{Tr}\left\{\mathbf{W}\mathbb{E}\left[\operatorname{vec}\left(\bm{\Delta}\right)\operatorname{vec}\left(\bm{\Delta}\right)^T\right]\right\} = \operatorname{Tr}\left\{\mathbf{W}\bm{\Sigma}_{\mathbf{v}}\left(\mathbf{U}\right)\right\},
    \end{align}
    where
    \begin{equation}
        \mathbf{W} = \left({{{\bar{\mathbf{J}} }^{-1}}}\right)^T\otimes\left({{{\bar{\mathbf{J}} }^{-1}}}\mathbf{T}{{{\bar{\mathbf{J}} }^{-1}}}\right)\succeq \mathbf{0}.
    \end{equation}
    
    We next show that $\mathbb{E}\left\{\operatorname{Tr}\left(\mathbf{T}\mathbf{R}_2\right)\right\} = \Theta\left(N^{-2}\right)$. To this end, note that
    \begin{equation}
        \operatorname{Tr}\left(\mathbf{T}\mathbf{R}_2\right) = \left\|\mathbf{T}_{\tau}\mathbf{Z}{{{\bar {\mathbf{J}} }^{-\frac{1}{2}}}}\right\|_F^2,
    \end{equation}
    Consequently,
    \begin{align}
       \lambda_{\min}\left({{{\bar {\mathbf{J}} }^{-1}}}\right) \left\|\mathbf{T}_{\tau}\mathbf{Z}\right\|_F^2 &\nonumber \le \|\mathbf{T}_{\tau}\mathbf{Z}{{{\bar {\mathbf{J}} }^{-\frac{1}{2}}}}\|_F^2\le \lambda_{\max}\left({{{\bar {\mathbf{J}} }^{-1}}}\right) \left\|\mathbf{T}_{\tau}\mathbf{Z}\right\|_F^2,
    \end{align}
    and it suffices to characterize the scaling behavior of the two bounding terms. As a direct consequence of Lemma \ref{EFIM_order}, the eigenvalues of ${{{\bar {\mathbf{J}} }^{-1}}}$ satisfy
        \begin{align}
            c_1^{-1}N^{-1}\le \lambda_{\min}\left({{{\bar {\mathbf{J}} }^{-1}}}\right)\le\lambda_{\max}\left({{{\bar {\mathbf{J}} }^{-1}}}\right) \le c_0^{-1}N^{-1}.
        \end{align}
    We now turn to the evaluation of $\mathbb{E}\left(\|\mathbf{T}_{\tau}\mathbf{Z}\|_F^2\right)$. Recall Lemma \ref{Z_HW_concerntration}. It is straightforward to observe
    \begin{align}
        &\mathbb{E}\left(\left\|\mathbf{T}_\tau\mathbf{Z}\right\|_F^2\right)\le \left\|\mathbf{T}_\tau\right\|_F^2\mathbb{E}\left(\left\|\mathbf{Z}\right\|_2^2\right) = O\left(N^{-1}\right).
    \end{align}
     Following \eqref{FIM_decomposition}, one may rewrite $\mathbf{T}_{\tau}\mathbf{Z}$ as
     \begin{equation}
         \mathbf{T}_{\tau}\mathbf{Z} = \sum_{n=1}^N \left(\tilde{s}_n-1\right)\mathbf{T}_{\tau}{{{\bar {\mathbf{J}} }^{-1}}}\mathbf{C}_n.
     \end{equation}
     By letting
     \begin{equation}
         \widetilde{\mathbf{C}}: = \left[\operatorname{vec}\left(\mathbf{T}_{\tau}{{{\bar {\mathbf{J}} }^{-1}}}\mathbf{C}_1\right),\ldots,\operatorname{vec}\left(\mathbf{T}_{\tau}{{{\bar {\mathbf{J}} }^{-1}}}\mathbf{C}_N\right)\right],
     \end{equation}
     we have
     \begin{equation}
          \mathbb{E}\left(\left\|\mathbf{T}_{\tau}\mathbf{Z}\right\|_F^2\right) =  \mathbb{E}\left(\left\|\operatorname{vec}\left(\mathbf{T}_{\tau}\mathbf{Z}\right)\right\|_2^2\right) = \operatorname{Tr}\left(\widetilde{\mathbf{C}}\bm{\Sigma}_{\tilde{\mathbf{s}}}\widetilde{\mathbf{C}}^H\right)
     \end{equation}
    Recall \eqref{sigma_s_expression} in Appendix \ref{proof_prop_bistochastic}. When $1<\kappa<2$, we have
    \begin{equation}
        \left(\kappa-1\right)\mathbf{I}_N\preceq \bm{\Sigma}_{\tilde{\mathbf{s}}}\preceq \mathbf{I}_N.
    \end{equation}
    It follows that
    \begin{equation}
        \operatorname{Tr}\left(\widetilde{\mathbf{C}}\bm{\Sigma}_{\tilde{\mathbf{s}}}\widetilde{\mathbf{C}}^H\right)\ge \left(\kappa-1\right)\|\widetilde{\mathbf{C}}\|_F^2.
    \end{equation}
    It remains to examine the order of $\|\widetilde{\mathbf{C}}\|_F^2$. Note that
    \begin{equation}
        \sum_{n=1}^N \mathbf{T}_{\tau}{{{\bar {\mathbf{J}} }^{-1}}}\mathbf{C}_n = \mathbf{T}_{\tau}{{{\bar {\mathbf{J}} }^{-1}}} \sum_{n=1}^N \mathbf{C}_n = \mathbf{T}_{\tau}{{{\bar {\mathbf{J}} }^{-1}}} \bar{\mathbf{J}} = \mathbf{T}_{\tau}.
    \end{equation}
    Then according to the Cauchy-Schwarz inequality
    \begin{align}
        \|\widetilde{\mathbf{C}}\|_F^2 &\nonumber = \sum_{n=1}^N \|\mathbf{T}_{\tau}{{{\bar {\mathbf{J}} }^{-1}}}\mathbf{C}_n\|_F^2 \ge \frac{1}{N}\left\|\sum_{n=1}^N \mathbf{T}_{\tau}{{{\bar {\mathbf{J}} }^{-1}}}\mathbf{C}_n\right\|_F^2 = \frac{L}{N}.
    \end{align}
    This indicates that $\mathbb{E}\left(\|\mathbf{T}_{\tau}\mathbf{Z}\|_F^2\right) = \Theta(N^{-1})$, and thereby $\mathbb{E}\left\{\operatorname{Tr}\left(\mathbf{T}\mathbf{R}_2\right)\right\} = \Theta\left(N^{-2}\right)$.

    \subsection{Higher-Order Term}
    In this subsection, we characterize the order of
    the higher-order term. Recall \eqref{R_3_original_formulation}. On $\Omega_N$, we have $\|\mathbf{Z}\|_2<\rho<1$, and hence
    \begin{equation}
        \left\|(\mathbf{I}+\mathbf{Z})^{-1}\right\|_2
        \le
        \frac{1}{1-\rho}.
    \end{equation}
    Therefore,
    \begin{align}
        \left|
        \operatorname{Tr}\left(\mathbf{T}\mathbf{R}_3\right)
        \right|
        &\le
        \|\mathbf{T}\|_F
        \left\|
        \mathbf{Z}^3(\mathbf{I}+\mathbf{Z})^{-1}\bar{\mathbf{J}}^{-1}
        \right\|_F \nonumber\\
        &\le
        \frac{L}{1-\rho}
        \|\mathbf{Z}\|_2^3
        \left\|\bar{\mathbf{J}}^{-1}\right\|_2,
    \end{align}
    By Lemma~\ref{EFIM_order}, $\|\bar{\mathbf{J}}^{-1}\|_2=O(N^{-1})$. Moreover, since $\Pr(\Omega_N)=1-O(e^{-cN})$, Lemma~\ref{Z_HW_concerntration} gives
    \begin{align}
        \mathbb{E}\left(\|\mathbf{Z}\|_2^3\mid\Omega_N\right)
        &\nonumber=
        \frac{\mathbb{E}\left(\|\mathbf{Z}\|_2^3\cdot 1_{\Omega_N}\right)}
        {\Pr(\Omega_N)}\le
        \frac{\mathbb{E}\left(\|\mathbf{Z}\|_2^3\right)}
        {\Pr(\Omega_N)}
        =
        O(N^{-3/2}).
    \end{align}
    Combining the above estimates yields
    \begin{align}
        \left|
        \mathbb{E}\left\{
        \operatorname{Tr}\left(\mathbf{T}\mathbf{R}_3\right)
        \mid\Omega_N
        \right\}
        \right|
        =
        O(N^{-1})O(N^{-3/2})
        =
        O(N^{-5/2}).
    \end{align}

    \subsection{CRB Gap between OFDM and Frequency-Spread Waveforms} 
    Combining the aforementioned results, the order of each term in \eqref{trace_expansion} can be evaluated as 
    \begin{align}
    &\mathbb{E}\left[f\left(\mathbf{U}\right)\left|\Omega_N\right.\right]\nonumber= \operatorname{Tr} \left( \mathbf{T}\bar{\mathbf{J}}^{-1} \right) - \underbrace{ \mathbb{E}\{ \operatorname{Tr}(\mathbf{T}\mathbf{R}_1) \left|\Omega_N\right.\} }_{O\left(N^{-3/2}e^{-cN}\right)} \\
    & \quad \quad\quad+ \underbrace{ \mathbb{E}\{ \operatorname{Tr}(\mathbf{T}\mathbf{R}_2) \left|\Omega_N\right.\} }_{\Theta(N^{-2})} - \underbrace{ \mathbb{E}\{ \operatorname{Tr}(\mathbf{T}\mathbf{R}_3) \left|\Omega_N\right.\} }_{O(N^{-5/2})}.
    \end{align}
    We next quantify the CRB gap between OFDM and frequency-spread waveforms. By viewing each term as a function of the waveform $\mathbf{U}$, we obtain 
    \begin{align}
    &\mathbb{E}\left[f\left(\mathbf{U}\right)\left|\Omega_N\right.\right] - \mathbb{E}\left[f\left(\mathbf{F}_N^H\right)\left|\Omega_N\right.\right] =\nonumber\\
    &  -\underbrace{ \left(\mathbb{E}\left\{\operatorname{Tr}\left[\mathbf{T}\mathbf{R}_1\left(\mathbf{U}\right)\left|\Omega_N\right.\right]\right\} - \mathbb{E}\left\{\operatorname{Tr}\left[\mathbf{T}\mathbf{R}_1\left(\mathbf{F}_N^H\right)\left|\Omega_N\right.\right]\right\}\right)}_{O\left(N^{-3/2}e^{-cN}\right)} \nonumber\\
    & + \underbrace{ \left(\mathbb{E}\left\{\operatorname{Tr}\left[\mathbf{T}\mathbf{R}_2\left(\mathbf{U}\right)\left|\Omega_N\right.\right]\right\} - \mathbb{E}\left\{\operatorname{Tr}\left[\mathbf{T}\mathbf{R}_2\left(\mathbf{F}_N^H\right)\left|\Omega_N\right.\right]\right\}\right)}_{O\left(N^{-2}\right)} \nonumber\\
    & - \underbrace{ \left(\mathbb{E}\left\{\operatorname{Tr}\left[\mathbf{T}\mathbf{R}_3\left(\mathbf{U}\right)\left|\Omega_N\right.\right]\right\} - \mathbb{E}\left\{\operatorname{Tr}\left[\mathbf{T}\mathbf{R}_3\left(\mathbf{F}_N^H\right)\left|\Omega_N\right.\right]\right\}\right)}_{O\left(N^{-5/2}\right)}.
    \end{align}
    Therefore, to show that OFDM achieves a lower CRB, it suffices to prove that the second-order gap is positive and equals to $\Theta(N^{-2})$. Since the conditional expectation is dominated by the total expectation $\mathbb{E}\left\{\operatorname{Tr}\left(\mathbf{T}\mathbf{R}_2\right)\right\}$, we focus on the gap in terms of $\mathbb{E}\left\{\operatorname{Tr}\left[\mathbf{T}\mathbf{R}_2\right]\right\}$, which is given as 
    \begin{align}\label{2nd_gap}
    &\nonumber\mathbb{E}\left\{\operatorname{Tr}\left[\mathbf{T}\mathbf{R}_2\left(\mathbf{U}\right)\right]\right\} - \mathbb{E}\left\{\operatorname{Tr}\left[\mathbf{T}\mathbf{R}_2\left(\mathbf{F}_N^H\right)\right]\right\}\\
    &\nonumber = \operatorname{Tr}\left\{\widetilde{\mathbf{H}}^T\mathbf{W}\widetilde{\mathbf{H}}\left[\bm{\Sigma}_{\tilde{\mathbf{s}}}\left(\mathbf{U}\right)-\bm{\Sigma}_{\tilde{\mathbf{s}}}\left(\mathbf{F}_N^H\right)\right]\right\}\\
    & = \left(2 - \kappa\right)\operatorname{Tr}\left[\widetilde{\mathbf{H}}^T\mathbf{W}\widetilde{\mathbf{H}}\left(\mathbf{I}_N - \mathbf{B}^T\mathbf{B}\right)\right],
    \end{align}
    where $\widetilde{\mathbf{H}}$ is defined in \eqref{tilde_Jacobian}, and $\mathbf{B}$
    is a doubly stochastic matrix defined in \eqref{sigma_s_expression}, whose rows and columns sum to one. For $n\ne m$, the off-diagonal entries satisfy 
    \begin{equation}
    \left[\mathbf{I}_N - \mathbf{B}^T\mathbf{B}\right]_{n,m} = - \mathbf{b}_n^T\mathbf{b}_m,\quad \forall n\ne m,
    \end{equation}
    while the diagonal entries can be expressed as 
    \begin{align}
    &\nonumber\left[\mathbf{I}_N - \mathbf{B}^T\mathbf{B}\right]_{n,n} = 1 -  \sum_{i=1}^N b_{i,n}^2 = \sum_{i=1}^Nb_{i,n}\left(1-b_{i,n}\right)\\
    & = \sum_{i=1}^N b_{i,n}\sum_{m\ne n}^N b_{i,m} = \sum_{m\ne n}^N\sum_{i=1}^N b_{i,n} b_{i,m}= \sum_{m\ne n}^N  \mathbf{b}_n^T\mathbf{b}_m.
    \end{align}
    Consequently, $\mathbf{I}_N-\mathbf{B}^T\mathbf{B}$ admits the Laplacian decomposition 
    \begin{equation}\label{Laplacian}
    \mathbf{I}_N - \mathbf{B}^T\mathbf{B} = \sum_{n<m}  \mathbf{b}_n^T\mathbf{b}_m\left(\mathbf{e}_n - \mathbf{e}_m\right)\left(\mathbf{e}_n - \mathbf{e}_m\right)^T,
    \end{equation}
    where $\mathbf{e}_n\in\mathbb{R}^N$ is the $n$-th column of $\mathbf{I}_N$. By defining
    \begin{equation}
        \mathbf{D}_{nm}:=\mathbf{C}_n - \mathbf{C}_m,
    \end{equation}
    and substituting \eqref{Laplacian} into \eqref{2nd_gap}, we obtain 
    \begin{align}\label{2nd_order_CRB_gap_closed_form}
    &\nonumber\mathbb{E}\left\{\operatorname{Tr}\left[\mathbf{T}\mathbf{R}_2\left(\mathbf{U}\right)\right]\right\} - \mathbb{E}\left\{\operatorname{Tr}\left[\mathbf{T}\mathbf{R}_2\left(\mathbf{F}_N^H\right)\right]\right\}\\
    &\nonumber =  \left(2 - \kappa\right)\sum_{n<m}  \mathbf{b}_n^T\mathbf{b}_m \left(\mathbf{e}_n - \mathbf{e}_m\right)^T\widetilde{\mathbf{H}}^T\mathbf{W}\widetilde{\mathbf{H}} \left(\mathbf{e}_n - \mathbf{e}_m\right)\\
    &\nonumber = \left(2 - \kappa\right)\sum_{n<m}  \mathbf{b}_n^T\mathbf{b}_m \operatorname{vec}\left(\mathbf{D}_{nm}\right)^T\mathbf{W}\operatorname{vec}\left(\mathbf{D}_{nm}\right)\\
    &\nonumber = \left(2 - \kappa\right)\sum_{n<m}  \mathbf{b}_n^T\mathbf{b}_m \operatorname{Tr}\left(\mathbf{T}{{{\bar{\mathbf{J}} }^{-1}}}\mathbf{D}_{nm}{{{\bar {\mathbf{J}} }^{-1}}}\mathbf{D}_{nm}{{{\bar {\mathbf{J}} }^{-1}}}\right)\\
    &  = \left(2 - \kappa\right)\sum_{n<m} \mu_{nm} \mathbf{b}_n^T\mathbf{b}_m \ge 0,
    \end{align}
    where $\mu_{nm}:=\|\mathbf{T}_{\tau}{{{\bar {\mathbf{J}} }^{-1}}}\mathbf{D}_{nm}{{{\bar {\mathbf{J}} }^{-1/2}}}\|_F^2$.
    
    Let us recall the definition of the RMS bandwidth of the $n$th basis for the $\alpha$-spread waveform $\mathbf{U}$, with the power of its $k$th subcarrier $f_k = \frac{k-1}{N}$ denoted as $p_{k,n}$. It follows that
    \begin{align}
    &\nonumber\alpha^2 \le \left|B_{\rm rms}(\mathbf{u}_n)\right|^2 = \sum_{k=1} |f_k-\bar f(\mathbf{u}_n)|^2 p_{k,n} \\
    & = \frac{1}{2}\sum_{k,\ell}(f_k - f_\ell)^2p_{k,n}p_{\ell,n} = \sum_{k<\ell} (f_k - f_\ell)^2p_{k,n}p_{\ell,n},
    \end{align}
    where $p_{k,n} = |[\mathbf{F}_N\mathbf{U}]_{k,n}|^2 = b_{n,k}$. This immediately yields 
    \begin{equation}\label{frequency_difference_weighted_bnbm}
    \sum_{n<m} (f_n - f_m)^2\sum_{r =1}^N p_{n,r}p_{m,r} =  \sum_{n<m} (f_n - f_m)^2 \mathbf{b}_n^T \mathbf{b}_m\ge \alpha^2 N.
    \end{equation}
    
    Moreover, by examining the entries in $\mathbf{C}_n$, we have
    \begin{align}
        &\nonumber\left[\mathbf{C}_n\right]_{\ell,L+\ell} = \frac{4\pi f_n}{\sigma^2}\operatorname{Im}\left(\beta_\ell\right),\quad \ell =1,2,\ldots, L,\\
        &\left[\mathbf{C}_n\right]_{\ell,2L+\ell} = \frac{4\pi f_n}{\sigma^2}\operatorname{Re}\left(\beta_\ell\right),\quad \ell =1,2,\ldots, L.
    \end{align}
    Therefore
    \begin{align}
        &\nonumber\left[\mathbf{D}_{nm}\right]_{\ell,L+\ell} = \frac{4\pi\operatorname{Im}\left(\beta_\ell\right)}{\sigma^2}\left(f_n-f_m\right),\quad \ell =1,2,\ldots, L,\\
        &\left[\mathbf{D}_{nm}\right]_{\ell,2L+\ell} = \frac{4\pi\operatorname{Re}\left(\beta_\ell\right)}{\sigma^2}\left(f_n-f_m\right),\quad \ell =1,2,\ldots, L,
    \end{align}
    which yields
    \begin{align}
        \left\|\mathbf{D}_{nm}\right\|_F^2 \nonumber&\ge 2\sum_{\ell = 1}^L\left(\left|\left[\mathbf{D}_{nm}\right]_{\ell,L+\ell}\right|^2+\left|\left[\mathbf{D}_{nm}\right]_{\ell,2L+\ell}\right|^2\right)\\
        & = \frac{32\pi^2}{\sigma^4} \sum_{\ell}|\beta_\ell|^2 |f_n-f_m|^2.
    \end{align}
    
    To establish the upper bound, note that the $(p,q)$-th entry of
    $\mathbf{D}_{nm}=\mathbf{C}_n-\mathbf{C}_m$ is a finite linear combination of
    differences of the following forms with $0\le f_n,f_m\le 1$:
    \begin{align}
        &\nonumber f_n^2 e^{-j2\pi f_n(\tau_q-\tau_p)}
        -
        f_m^2 e^{-j2\pi f_m(\tau_q-\tau_p)}=
        O\left(|f_n-f_m|\right),\\
        &\nonumber f_n e^{-j2\pi f_n(\tau_q-\tau_p)}
        -
        f_m e^{-j2\pi f_m(\tau_q-\tau_p)}=
        O\left(|f_n-f_m|\right),\\
        &e^{-j2\pi f_n(\tau_q-\tau_p)}
        -
        e^{-j2\pi f_m(\tau_q-\tau_p)}=
        O\left(|f_n-f_m|\right).
    \end{align}
    Therefore,
    \begin{equation}
        [\mathbf{D}_{nm}]_{p,q}
        =
        O\left(|f_n-f_m|\right),\quad 1\le p,q\le 3L .
    \end{equation}
    Since $L$ is fixed, it follows that
    \begin{align}
        \left\|\mathbf{D}_{nm}\right\|_F^2
        &=
        \sum_{p=1}^{3L}\sum_{q=1}^{3L}
        \left|[\mathbf{D}_{nm}]_{p,q}\right|^2 =
        O\left(|f_n-f_m|^2\right).
    \end{align}
    Combining both bounds gives
    \begin{equation}\label{Dnm_order}
        \|\mathbf{D}_{nm}\|_F^2
        =
        \Theta\left(|f_n-f_m|^2\right),
    \end{equation}
    which implies that
    \begin{equation}
        \mu_{nm} =\|\mathbf{T}_{\tau}{{{\bar {\mathbf{J}} }^{-1}}}\mathbf{D}_{nm}{{{\bar {\mathbf{J}} }^{-1/2}}}\|_F^2 = \Theta\left(N^{-3}|f_n-f_m|^2\right).
    \end{equation}
    Together with \eqref{frequency_difference_weighted_bnbm}, this suggests that
    \begin{align}\label{R2_order}
    &\nonumber\mathbb{E}\left\{\operatorname{Tr}\left[\mathbf{T}\mathbf{R}_2\left(\mathbf{U}\right)\right]\right\} - \mathbb{E}\left\{\operatorname{Tr}\left[\mathbf{T}\mathbf{R}_2\left(\mathbf{F}_N^H\right)\right]\right\}\\
    & =  \left(2 - \kappa\right)\sum_{n<m} \mu_{nm} \mathbf{b}_n^T\mathbf{b}_m = \Theta(N^{-2}) >0.
    \end{align}
    Therefore, the CRB gap between OFDM and any $\alpha$-frequency-spread waveform is dominated by the second-order gap, which is strictly positive. As a result, with sufficiently large $N$, OFDM achieves lower ranging CRB compared to frequency-spread waveforms. 

\section{Proof of Theorem \ref{stationarity_thm}}\label{proof_thm_2}
    Consider a geodesic on $\mathbb{U}(N)$ passing through $\mathbf{U}=\mathbf{F}_N^H$, parameterized by~\cite{absil2008optimization}
    \begin{equation}
        \mathbf{U}(t)=\mathbf{F}_N^H e^{t\mathbf{K}}, \quad \mathbf{U}(0)=\mathbf{F}_N^H,
    \end{equation}
    where $\mathbf{K}=-\mathbf{K}^H$ is a skew-Hermitian matrix satisfying $\|\mathbf{K}\|_2=1$. Under this parameterization,
    \begin{equation}\label{Qt_form}
        \mathbf{Q}(t)=e^{-t\mathbf{K}}, \quad \mathbf{Q}(t)^H\mathbf{s}=e^{t\mathbf{K}}\mathbf{s}.
    \end{equation}
    Let $\mathbf{J}(t):=\mathbf{J}(\mathbf{U}(t),\mathbf{s})$, $\bm{\Delta}(t):=\mathbf{J}(t)-\bar{\mathbf{J}}$, and $\mathbf{Z}(t):=\bar{\mathbf{J}}^{-1}\bm{\Delta}(t)$. The FIM can be written as
    \begin{align}\label{FIM_expansion}
        \nonumber
        \mathbf{J}(t)
        &=\frac{2}{\sigma^2}\operatorname{Re}\left\{
        \mathbf{H}^H\operatorname{Diag}\left(\left|e^{t\mathbf{K}}\mathbf{s}\right|^2\right)\mathbf{H}
        \right\}\\
        & =
        \sum_{n=1}^N
        \left|\left[e^{t\mathbf{K}}\mathbf{s}\right]_n\right|^2\mathbf{C}_n,
    \end{align}
    where $\mathbf{C}_n$ is defined in \eqref{Cn_def}. For notational brevity, let $f(t): = f(\mathbf{U}(t))$. Our objective is to demonstrate that
    \begin{equation}
        \mathbb{E}\left[\dot{f}(0)\right] = \left. \frac{d\mathbb{E}\left[{f}(t)\right]}{dt} \right|_{t=0} = 0
    \end{equation}
    holds for any skew-Hermitian matrix $\mathbf{K}$.

    Applying the matrix inverse derivative identity, we have
    \begin{equation}
        \dot{\mathbf{J}}^{-1}(t) = -\mathbf{J}^{-1}(t) \dot{\mathbf{J}}(t) \mathbf{J}^{-1}(t).
    \end{equation}
    Consequently, the derivative of $\mathbb{E}\left[{f}(t)\right]$ at $t=0$ is given by
    \begin{align}
        \mathbb{E}\left[\dot{f}(0)\right] &= \left. \frac{d}{dt} \operatorname{Tr} \left\{ \mathbb{E} \left[\mathbf{T}\mathbf{J}^{-1}(t) \right] \right\} \right|_{t=0} \nonumber \\
        &= -\mathbb{E} \left\{ \operatorname{Tr} \left[\mathbf{T}\mathbf{J}^{-1}(0) \dot{\mathbf{J}}(0) \mathbf{J}^{-1}(0) \right] \right\}.
    \end{align}
    Define
    \begin{equation}\label{gt_def}
        g_n(t)=\left|\left[e^{t\mathbf{K}}\mathbf{s}\right]_n\right|^2-1,
    \end{equation}
    the FIM may be expressed as
    \begin{equation}
        \mathbf{J}(t) = \sum_{n=1}^N \left[g_n(t)+1\right] \mathbf{C}_n,
    \end{equation}
    The derivative of the FIM at $t=0$ is therefore
    \begin{equation}
        \dot{\mathbf{J}}(0) = \sum_{n=1}^N \dot{g}_n(0) \mathbf{C}_n,
    \end{equation}
    which yields
    \begin{equation}
        \mathbb{E}\left[\dot{f}(0)\right] = -\sum_{n=1}^N \mathbb{E} \{ \dot{g}_n(0) \gamma_n(\mathbf{s}) \},
    \end{equation}
    where
    \begin{equation}
        \gamma_n(\mathbf{s}) := \operatorname{Tr} \left[ \mathbf{T}\mathbf{J}^{-1}(0) \mathbf{C}_n \mathbf{J}^{-1}(0) \right].
    \end{equation}

    Applying the Taylor expansion of the matrix exponential, we have
    \begin{equation}
        e^{t\mathbf{K}}\mathbf{s} = \mathbf{s}+t\mathbf{K}\mathbf{s}+o(t).
    \end{equation}
    It follows that 
    \begin{equation}
        |[e^{t\mathbf{K}}\mathbf{s}]_n|^2 = |s_n|^2 + 2t\operatorname{Re}(s_n^\ast[\mathbf{K}\mathbf{s}]_n) + o(t),
    \end{equation}
    and the derivative at $t=0$ is given by
    \begin{align}
        \dot{g}_n(0) &= 2\operatorname{Re}(s_n^\ast[\mathbf{K}\mathbf{s}]_n) = 2\operatorname{Re}\left(\sum_{m=1}^N[\mathbf{K}]_{n,m}s_m s_n^\ast\right) \nonumber \\
        &= 2\operatorname{Re}([\mathbf{K}]_{n,n}|s_n|^2) + 2\operatorname{Re}\left(\sum_{\substack{m = 1 \\ m \ne n}}^N [\mathbf{K}]_{n,m}s_m s_n^\ast\right).
    \end{align}
    Since $\mathbf{K}$ is skew-Hermitian, its diagonal elements $[\mathbf{K}]_{n,n}$ are purely imaginary, which implies $\operatorname{Re}([\mathbf{K}]_{n,n}|s_n|^2) = 0$.
    
    By the law of total expectation, we have
    \begin{equation}
        \mathbb{E}\left[\dot{g}_n(0)\gamma_n(\mathbf{s})\right] = \mathbb{E}\left\{\mathbb{E}\left[\dot{g}_n(0)\gamma_n(\mathbf{s}) \left|\left|\mathbf{s}\right|^2\right.\right]\right\}.
    \end{equation}
    As $\mathbf{J}\left(0\right)$ depends on the symbols $\mathbf{s}$ only through their squared magnitudes $\left|\mathbf{s}\right|^2$, the same applies to $\gamma_n(\mathbf{s})$, which remains constant given $\left|\mathbf{s}\right|^2$. Consequently,
    \begin{align}
        &\mathbb{E}\left[\dot{g}_n(0)\gamma_n(\mathbf{s}) \left|\left|\mathbf{s}\right|^2\right.\right] = \gamma_n(\mathbf{s}) \mathbb{E}\left[\dot{g}_n(0)\left|\left|\mathbf{s}\right|^2\right.\right] \nonumber \\
        &\nonumber= 2\gamma_n(\mathbf{s}) \operatorname{Re} \left\{ \sum_{\substack{m = 1 \\ m \ne n}}^N [\mathbf{K}]_{n,m} \mathbb{E}\left(s_m s_n^\ast\left|\left|\mathbf{s}\right|^2\right.\right)\right\}\\
        &\nonumber=2\gamma_n(\mathbf{s}) \operatorname{Re} \left\{ \sum_{\substack{m = 1 \\ m \ne n}}^N [\mathbf{K}]_{n,m} \mathbb{E}\left(s_m \left|\left|s_m\right|^2\right.\right)\mathbb{E}\left(s_n^\ast\left|\left|s_n\right|^2\right.\right)\right\}\\
        & = 0.
    \end{align}
    where the last equality holds from Assumption \ref{assumption3} due to the central symmetry of the constellation. This leads to $\mathbb{E}[\dot{f}(0)] = 0$ for any skew-Hermitian matrix $\mathbf{K}$, thereby completing the proof.

    \section{Proof of Theorem \ref{hessian_thm}}\label{proof_thm3}
    \subsection{General Framework}
    Let us recall the resolvent expansion of $\mathbf{J}^{-1}$ in \eqref{J_expansion}, and define $\mathbf{R}_1(t)$, $\mathbf{R}_2(t)$, and $\mathbf{R}_3(t)$ by replacing $\mathbf{U}$ with $\mathbf{U}(t)$. The Riemannian Hessian may thus be decomposed as
    \begin{align}\label{Hessian_decomposition_total}
        &\nonumber\frac{d^{2}}{dt^{2}}\mathbb{E}\left\{f\left[\mathbf{U}(t)\right]\right\} =-\frac{d^{2}}{dt^{2}}\mathbb{E}\left\{\operatorname{Tr}\left[\mathbf{T}\mathbf{R}_1(t)\right]\right\}\\&+ \frac{d^{2}}{dt^{2}}\mathbb{E}\left\{\operatorname{Tr}\left[\mathbf{T}\mathbf{R}_2(t)\right]\right\}-\frac{d^{2}}{dt^{2}}\mathbb{E}\left\{\operatorname{Tr}\left[\mathbf{T}\mathbf{R}_3(t)\right]\right\}.
    \end{align}
    Under OFDM, i.e., at $t=0$, the FIM is invertible for every symbol realization because the constellation alphabet contains no zero element. Since the constellation alphabet is finite, there exist constants $s_{\min},s_{\max}>0$ such that
    \begin{equation}
        s_{\min}\le |s_n|\le s_{\max},\quad \forall n .
    \end{equation}
    At $t=0$, we have
    \begin{equation}
        \mathbf J(0)=\sum_{n=1}^N |s_n|^2\mathbf C_n \succeq s_{\min}^2 \sum_{n=1}^N\mathbf C_n = s_{\min}^2\bar{\mathbf J}.
    \end{equation}
    Together with $\|\bar{\mathbf J}^{-1}\|_2=O(N^{-1})$, this implies
    \begin{equation}
        \|\mathbf J^{-1}(0)\|_2=O(N^{-1}).
    \end{equation}
    Next, note that
    \begin{equation}
        \dot{\mathbf J}(0)
        =
        \sum_{n=1}^N \dot g_n(0)\mathbf C_n,\qquad
        \ddot{\mathbf J}(0)
        =
        \sum_{n=1}^N \ddot g_n(0)\mathbf C_n .
    \end{equation}
    Since $\|\mathbf K\|_2=1$ and $\|\mathbf s\|_2=O(N^{1/2})$, we have
    \begin{equation}
        |[\mathbf K\mathbf s]_n|\le \|\mathbf K\mathbf s\|_2=O(N^{1/2}),
    \end{equation}
    and similarly
    \begin{equation}
        |[\mathbf K^2\mathbf s]_n|\le \|\mathbf K^2\mathbf s\|_2=O(N^{1/2}).
    \end{equation}
    Therefore,
    \begin{equation}
        |\dot g_n(0)|=\left|2\operatorname{Re}\left(s_n^\ast[\mathbf K\mathbf s]_n\right)\right|=O(N^{1/2}),
    \end{equation}
    and
    \begin{align}
        \nonumber|\ddot g_n(0)|&=2\left|[\mathbf{K}\mathbf{s}]_n\right|^2
        +
        2\operatorname{Re}\left(s_n^\ast[\mathbf{K}^2\mathbf{s}]_n\right)\\
        &\le
        2\left|[\mathbf K\mathbf s]_n\right|^2+2|s_n|\left|[\mathbf K^2\mathbf s]_n\right| =O(N).
    \end{align}
    Since
    \begin{equation}
        \|\mathbf{C}_n\|_2
        \le
        \frac{2}{\sigma^2}\left\|\widetilde{\mathbf{h}}_{n}\widetilde{\mathbf{h}}_{n}^H\right\|_2
        =
        \frac{2}{\sigma^2}\|\widetilde{\mathbf{h}}_{n}\|_2^2
        =
        O(1),
    \end{equation}
    we obtain
    \begin{equation}
        \|\dot{\mathbf J}(0)\|_2=O(N^{3/2}),\qquad \|\ddot{\mathbf J}(0)\|_2=O(N^2).
    \end{equation}
    Moreover, the second derivative of $\mathbf{J}^{-1}(0)$ is
    \begin{align}
        &\nonumber\ddot{\mathbf J}^{-1}(0)
        =\\
        &2\mathbf J^{-1}(0)\dot{\mathbf J}(0)
        \mathbf J^{-1}(0)\dot{\mathbf J}(0)\mathbf J^{-1}(0)-\mathbf J^{-1}(0)\ddot{\mathbf J}(0)\mathbf J^{-1}(0).
    \end{align}
    Consequently,
    \begin{align}
        \left\|\ddot{\mathbf J}^{-1}(0)\right\|_2
        &\le
        2O(N^{-1})O(N^{3/2})O(N^{-1})
        O(N^{3/2})O(N^{-1}) \nonumber\\
        &\quad
        +O(N^{-1})O(N^2)O(N^{-1})=O(1).
    \end{align}
    It follows that
    \begin{equation}
        \operatorname{Tr}\left[\mathbf T\ddot{\mathbf J}^{-1}(0)\right]=
        O(1).
    \end{equation}
    Therefore, using $\Pr(\Omega_N^c)=O(e^{-cN})$, we have
    \begin{align}
        \left.
        \frac{d^2}{dt^2}
        \mathbb E\left\{
        \operatorname{Tr}\left[\mathbf T\mathbf J^{-1}(t)\right]\cdot
        1_{\Omega_N^c}
        \right\}
        \right|_{t=0}
        =
        O(e^{-cN}).
    \end{align}
    Consequently, the total expected Hessian and its restriction to $\Omega_N$ differ only by an exponentially small term:
    \begin{align}
        \left.
        \frac{d^2}{dt^2}
        \mathbb E\left\{
        f[\mathbf U(t)]
        \right\}
        \right|_{t=0}=
        \left.
        \frac{d^2}{dt^2}
        \mathbb E\left\{f[\mathbf U(t)]
        \cdot 1_{\Omega_N}
        \right\}
        \right|_{t=0}
        +
        O(e^{-cN}).
    \end{align}
    Moreover, the same complement estimate applies to the finite-order terms $\mathbf R_1(t)$ and $\mathbf R_2(t)$. Indeed, their second derivatives contain only finitely many products of $\bar{\mathbf J}^{-1}$, $\bm\Delta(t)$, $\dot{\bm\Delta}(t)$, and $\ddot{\bm\Delta}(t)$ when $t = 0$, and are therefore polynomially bounded in $N$ under the finite-alphabet assumption. Hence,
    \begin{equation}
        \left.
        \frac{d^2}{dt^2}
        \mathbb E\left\{
        \operatorname{Tr}\left[\mathbf T\mathbf R_i(t)\right]\cdot
        1_{\Omega_N^c}
        \right\}
        \right|_{t=0}
        =
        O(e^{-cN}),\quad i=1,2.
    \end{equation}
    Therefore, one may analyze the total expectations of the first- and second-order terms, while restricting the higher-order term to the almost-sure set $\Omega_N$.
    
    \subsection{First-Order Term}
    Since $\mathbb{E}\left[\bm{\Delta}(t)\right]=\mathbf{0}$, we have
    \begin{align}
        \mathbb E\left\{
        \operatorname{Tr}\left[\mathbf T\mathbf R_1(t)\right]
        \right\}
        =
        \operatorname{Tr}\left[
        \mathbf T\bar{\mathbf J}^{-1}
        \mathbb E[\bm\Delta(t)]
        \bar{\mathbf J}^{-1}
        \right]=0,
    \end{align}
    which immediately gives
    \begin{equation}
        \left.
        \frac{d^2}{dt^2}
        \mathbb E\left\{
        \operatorname{Tr}\left[\mathbf T\mathbf R_1(t)\right]
        \right\}
        \right|_{t=0}
        =
        0.
    \end{equation}
    
    \subsection{Second-Order Term}
    Recall \eqref{2nd_term_analysis} and \eqref{covaraince_Jacobian}. We have
    \begin{align}\label{R2_Hessian}
        & \nonumber \frac{d^{2}}{dt^{2}}\mathbb{E}\left\{\operatorname{Tr}\left[\mathbf{T}\mathbf{R}_2(t)\right]\right\} = \frac{d^{2}}{dt^{2}}\operatorname{Tr}\left\{\widetilde{\mathbf{H}}^T\mathbf{W}\widetilde{\mathbf{H}}\bm{\Sigma}_{\tilde{\mathbf{s}}}\left(t\right)\right\}\\
        & = \frac{d^{2}}{dt^{2}}\operatorname{Tr}\left\{\widetilde{\mathbf{H}}^T\mathbf{W}\widetilde{\mathbf{H}}\left[\mathbf{I}_N+\left(\kappa-2\right)\mathbf{B}^T\left(t\right)\mathbf{B}\left(t\right)\right]\right\},
    \end{align}
    where $\bm{\Sigma}_{\tilde{\mathbf{s}}}(t)$ and $\mathbf{B}(t)$ are defined by replacing $\mathbf{U}$ in $\bm{\Sigma}_{\tilde{\mathbf{s}}}$ and $\mathbf{B}$ with $\mathbf{U}(t)$. Recalling \eqref{Qt_form},
    \begin{equation}
        \mathbf{Q}\left(t\right) = e^{-t\mathbf{K}} = \mathbf{I}_N - t\mathbf{K}+\frac{t^2}{2}\mathbf{K}^2 + O(t^3).
    \end{equation}
    Thus,
    \begin{equation}
        \left[\mathbf{Q}\left(t\right)\right]_{m,n} = \delta_{m,n} - t\left[\mathbf{K}\right]_{m,n} + \frac{t^2}{2}\left[\mathbf{K}^2\right]_{m,n} + O(t^3),
    \end{equation}
    which yields
    \begin{align}
        \left[\mathbf{B}\left(t\right)\right]_{m,n} &=\left|\left[\mathbf{Q}\left(t\right)\right]_{m,n}\right|^2\nonumber \\
        &= \begin{cases}
            t^2\left|\left[\mathbf{K}\right]_{m,n}\right|^2 +O(t^3), & m\ne n, \\[1mm]
          1 -t^2\sum_{\ell\ne n} \left|\left[\mathbf{K}\right]_{\ell,n}\right|^2  + O(t^3), & m = n.
        \end{cases}
    \end{align}
    Accordingly,
    \begin{equation}
        \mathbf{B}\left(0\right) = \mathbf{I}_N,\quad\dot{\mathbf{B}}\left(0\right) = \mathbf{0},
    \end{equation}
    and
    \begin{align}
        \left[\ddot{\mathbf{B}}\left(0\right)\right]_{m,n} = \begin{cases}
            2\left|\left[\mathbf{K}\right]_{m,n}\right|^2, & m\ne n,\\[1mm]
          -2\sum_{\ell\ne n} \left|\left[\mathbf{K}\right]_{\ell,n}\right|^2 , & m = n.
        \end{cases}
    \end{align}
    Therefore,
    \begin{align}\label{I_BB_Hessian}
        &\nonumber\left.\frac{d^{2}}{dt^{2}}\left[\mathbf{I}_N -\mathbf{B}^T\left(t\right)\mathbf{B}\left(t\right)\right]\right|_{t=0} \\
        &\nonumber= -\ddot{\mathbf{B}}^T\left(0\right){\mathbf{B}}\left(0\right) - 2 \dot{\mathbf{B}}^T\left(0\right)\dot{\mathbf{B}}\left(0\right) - {\mathbf{B}}^T\left(0\right)\ddot{\mathbf{B}}\left(0\right)\\
        & = -\ddot{\mathbf{B}}^T\left(0\right) - \ddot{\mathbf{B}}\left(0\right) = -2\ddot{\mathbf{B}}\left(0\right)\nonumber \\
        & = 4\sum_{n<m}\left|\left[\mathbf{K}\right]_{m,n}\right|^2\left(\mathbf{e}_n - \mathbf{e}_m\right)\left(\mathbf{e}_n - \mathbf{e}_m\right)^T.
    \end{align}
    Substituting \eqref{I_BB_Hessian} into \eqref{R2_Hessian} yields
    \begin{align}
         &\nonumber\left.\frac{d^{2}}{dt^{2}}\mathbb{E}\left\{\operatorname{Tr}\left[\mathbf{T}\mathbf{R}_2(t)\right]\right\}\right|_{t=0}\\
         &= 4\left(2-\kappa\right) \sum_{n<m}  \left|\left[\mathbf{K}\right]_{n,m}\right|^2 \operatorname{Tr}\left(\mathbf{T}{{{\bar{\mathbf{J}} }^{-1}}}\mathbf{D}_{nm}{{{\bar {\mathbf{J}} }^{-1}}}\mathbf{D}_{nm}{{{\bar {\mathbf{J}} }^{-1}}}\right)\nonumber \\
         & = 4\left(2-\kappa\right) \sum_{n<m}  \left|\left[\mathbf{K}\right]_{n,m}\right|^2 \mu_{nm}.
    \end{align}
    By a similar argument as in \eqref{R2_order}, we obtain
    \begin{align}
         &\left.\frac{d^{2}}{dt^{2}}\mathbb{E}\left\{\operatorname{Tr}\left[\mathbf{T}\mathbf{R}_2(t)\right]\right\}\right|_{t=0} = \Theta\left(N^{-3}C_{\mathbf{K}}\right)\ge0,
    \end{align}
    where
    \begin{align} \label{CK_DEF}
        C_{\mathbf{K}}: =\sum_{n<m} \left|\left[\mathbf{K}\right]_{n,m}\right|^2(f_n - f_m)^2.
    \end{align}
    
    \subsection{Higher-Order Term}
    We now analyze the higher-order term over the almost-sure set
    $\Omega_N=\{\|\mathbf Z\|_2<\rho\}$, where $0<\rho<1$. Define $\mathbf{A}_n=\bar{\mathbf{J}}^{-1}\mathbf{C}_n$. Then,
    \begin{equation}
        \mathbf{Z}(t)=\sum_{n=1}^N g_n(t)\mathbf{A}_n,
    \end{equation}
    where $g_n(t)$ is defined in \eqref{gt_def}. In particular,
    \begin{align}
        &\nonumber
        \mathbf{Z}(0)=\sum_{n=1}^N g_n(0)\mathbf{A}_n,\quad
        \dot{\mathbf{Z}}(0)=\sum_{n=1}^N \dot g_n(0)\mathbf{A}_n,\\
        &
        \ddot{\mathbf{Z}}(0)=\sum_{n=1}^N \ddot g_n(0)\mathbf{A}_n.
    \end{align}
    Since $ \|\mathbf{C}_n\|_2 = O(1)$, Lemma~\ref{EFIM_order} gives
    \begin{equation}
        \|\mathbf{A}_n\|_2
        =
        \|\bar{\mathbf{J}}^{-1}\mathbf{C}_n\|_2
        \le
        \|\bar{\mathbf{J}}^{-1}\|_2\|\mathbf{C}_n\|_2
        =
        O(N^{-1}).
    \end{equation}
    Moreover, note that
    \begin{align}
        &\nonumber
        g_n(0)=|s_n|^2-1,\\
        &\nonumber
        \dot g_n(0)=2\operatorname{Re}\left(s_n^\ast[\mathbf{K}\mathbf{s}]_n\right) =2\operatorname{Re}\left(\sum_{m}s_n^\ast[\mathbf{K}]_{n,m}s_m \right),\\
        &
        \ddot g_n(0)
        =
        2\left|[\mathbf{K}\mathbf{s}]_n\right|^2
        +
        2\operatorname{Re}\left(s_n^\ast[\mathbf{K}^2\mathbf{s}]_n\right).
    \end{align}
    Based on the fact $[\mathbf{K}]_{m,n} = -[\mathbf{K}]_{n,m}^\ast$, it follows that
    \begin{align}\label{pairwise_dotZ}
        \nonumber\dot{\mathbf{Z}}(0) &= 2\sum_{n<m}\operatorname{Re}\left(s_n^\ast[\mathbf{K}]_{n,m}s_m \right)\left(\mathbf{A}_n- \mathbf{A}_m\right)\\
        & = 2\sum_{n<m}\xi_{nm}\bar{\mathbf{J}}^{-1}\mathbf{D}_{nm},
    \end{align}
    where $\xi_{nm}: =\operatorname{Re}\left(s_n^\ast[\mathbf{K}]_{n,m}s_m \right)$.

    On $\Omega_N$, the Neumann expansion of $\mathbf R_3(t)$ is valid. Thus,
    \begin{equation}
        \mathbf{R}_3(t)
        =
        \mathbf{Z}^3(t)
        \left[\mathbf{I}+\mathbf{Z}(t)\right]^{-1}
        \bar{\mathbf{J}}^{-1}
        =
        \sum_{r=3}^{\infty}
        (-1)^{r-3}\mathbf{Z}^r(t)\bar{\mathbf{J}}^{-1}.
    \end{equation}
    Consequently,
    \begin{align}\label{R_3_Hessian_expansion}
        &\nonumber
        \left.
        \frac{d^2}{dt^2}
        \mathbb{E}\left\{
        \operatorname{Tr}\left[\mathbf T\mathbf R_3(t)\right]
        \cdot 1_{\Omega_N}
        \right\}
        \right|_{t=0}\\
        &=
        \sum_{r=3}^{\infty}
        (-1)^{r-3}
        \operatorname{Tr}\left[
        \mathbf T
        \mathbb{E}\left[
        \ddot{\mathbf Z}^{r}(0)\cdot 1_{\Omega_N}
        \right]
        \bar{\mathbf J}^{-1}
        \right].
    \end{align}
    
    For notational simplicity, we omit ``$(0)$'' in the following and write $\mathbf Z$, $\dot{\mathbf Z}$, $\ddot{\mathbf Z}$, $g_n$, $\dot{g}_n$ and $\ddot{g}_n$ for their values at $t=0$. We obtain
    \begin{align}\label{ddot_Z_r}
        &\nonumber
        \ddot{\mathbf{Z}}^r
        :=
        \left.\frac{d^2}{dt^2}\mathbf{Z}^r(t)\right|_{t=0}\\
        &=
        2\sum_{1\le \ell < k\le r}
        \mathbf{Z}^{\ell-1}\dot{\mathbf{Z}}\mathbf{Z}^{k-\ell-1}
        \dot{\mathbf{Z}}\mathbf{Z}^{r-k} + \sum_{k=1}^r
        \mathbf{Z}^{k-1}\ddot{\mathbf{Z}}\mathbf{Z}^{r-k}.
    \end{align}
    Thus, each term in \eqref{ddot_Z_r} takes one of the following two forms:
    \begin{align}
       &\text{Form 1:}\quad \mathbf{Z}^a\dot{\mathbf{Z}}\mathbf{Z}^b\dot{\mathbf{Z}}\mathbf{Z}^c,\quad a+b+c=r-2,
         \label{1st_form}\\
       &\text{Form 2:}\quad  \mathbf{Z}^a\ddot{\mathbf{Z}}\mathbf{Z}^b,\qquad\;\;\; a+b=r-1.
        \label{2nd_form}
    \end{align}
    
    We first analyze Form~1. Since $\|\mathbf{Z}\|_2<\rho$ on $\Omega_N$,
    \begin{equation}
        \left\|
        \mathbb{E}\left[
        \mathbf Z^a\dot{\mathbf Z}\mathbf Z^b
        \dot{\mathbf Z}\mathbf Z^c
        \cdot 1_{\Omega_N}
        \right]
        \right\|_2
        \le \rho^{r-3} \left\|
        \mathbb{E}\left[
        \mathbf Z\dot{\mathbf Z}
        \dot{\mathbf Z}
        \cdot 1_{\Omega_N}
        \right]
        \right\|_2,
    \end{equation}
    By using the pairwise representation in \eqref{pairwise_dotZ}, we have
    \begin{align}\label{form2_expectation_expansion}
        &\nonumber
        \mathbb{E}\left(
        \mathbf{Z}\dot{\mathbf{Z}}
        \dot{\mathbf{Z}}
        \right)=
        4\sum_{n<m}\sum_{p<q}
        \mathbb{E}\left(
        \xi_{nm}\xi_{pq}
        \mathbf{Z}
        \bar{\mathbf{J}}^{-1}\mathbf{D}_{nm}
        \bar{\mathbf{J}}^{-1}\mathbf{D}_{pq}
        \right).
    \end{align}
    Note that $\mathbf{Z}$ depends on $\mathbf{s}$ only through $|\mathbf{s}|^2$. For notational convenience, write
    \begin{equation}
        \eta_{nm}:=s_n^\ast[\mathbf{K}]_{n,m}s_m,\qquad n<m ,
    \end{equation}
    so that $\xi_{nm}=\operatorname{Re}(\eta_{nm})$. Then
    \begin{equation}
        \xi_{nm}\xi_{pq}
        =
        \frac{1}{4}
        \left(\eta_{nm}+\eta_{nm}^{\ast}\right)
        \left(\eta_{pq}+\eta_{pq}^{\ast}\right).
    \end{equation}
    If $(p,q)\ne(n,m)$, then in each of the four products
    \begin{align}
        \eta_{nm}\eta_{pq},\quad
        \eta_{nm}\eta_{pq}^{\ast},\quad
        \eta_{nm}^{\ast}\eta_{pq},\quad
        \eta_{nm}^{\ast}\eta_{pq}^{\ast},
    \end{align}
    there exists at least one symbol whose phase is unbalanced. By the conditional
    symmetry of Assumption \ref{assumption3}, all such terms vanish after
    conditioning on $|\mathbf{s}|^2$. Hence,
    \begin{equation}\label{xi_cross_zero}
        \mathbb{E}\left(
        \xi_{nm}\xi_{pq}
        \left|\left|\mathbf{s}\right|^2\right.
        \right)
        =
        0,\qquad (p,q)\ne(n,m).
    \end{equation}
    Thus only the terms with $(p,q)=(n,m)$ survive after conditional averaging. We obtain
    \begin{align}
        &\nonumber\mathbb{E}\left(
        \mathbf{Z}\dot{\mathbf{Z}}
        \dot{\mathbf{Z}}
        \right) \\
        &\nonumber = 4\sum_{n<m}\sum_{p<q}
        \mathbb{E}\left[\mathbb{E}\left(
        \xi_{nm}\xi_{pq}
        \mathbf{Z}
        \bar{\mathbf{J}}^{-1}\mathbf{D}_{nm}
        \bar{\mathbf{J}}^{-1}\mathbf{D}_{pq}\left||\mathbf{s}|^2\right.
        \right)\right]\\
        & \nonumber = 4\sum_{n<m}\sum_{p<q}
        \mathbb{E}\left[\mathbb{E}\left(
        \xi_{nm}\xi_{pq}
        \left||\mathbf{s}|^2\right.
        \right)\mathbf{Z}
        \bar{\mathbf{J}}^{-1}\mathbf{D}_{nm}
        \bar{\mathbf{J}}^{-1}\mathbf{D}_{pq}\right] \\
        & \nonumber = 4\sum_{n<m}
        \mathbb{E}\left[
        \mathbb{E}\left(\xi_{nm}^2\left||\mathbf{s}|^2\right.\right)
        \mathbf{Z}
        \bar{\mathbf{J}}^{-1}\mathbf{D}_{nm}
        \bar{\mathbf{J}}^{-1}\mathbf{D}_{nm}
        \right]\\
        & \nonumber   = 4\sum_{n<m}
        \mathbb{E}\left[
        \mathbb{E}\left(\xi_{nm}^2\mathbf{Z}\left||\mathbf{s}|^2\right.\right)
        \bar{\mathbf{J}}^{-1}\mathbf{D}_{nm}
        \bar{\mathbf{J}}^{-1}\mathbf{D}_{nm}
        \right]\\
        &\nonumber  = 4\sum_{n<m}\sum_{\ell}
        \mathbb{E}\left[
        \mathbb{E}\left(\xi_{nm}^2g_{\ell}\mathbf{A}_{\ell}\left||\mathbf{s}|^2\right.\right)
        \bar{\mathbf{J}}^{-1}\mathbf{D}_{nm}
        \bar{\mathbf{J}}^{-1}\mathbf{D}_{nm}
        \right]\\
        & = 4\sum_{n<m}\sum_{\ell}
        \mathbb{E}\left(\xi_{nm}^2g_{\ell}\right)
        \mathbf{A}_{\ell}
        \bar{\mathbf{J}}^{-1}\mathbf{D}_{nm}
        \bar{\mathbf{J}}^{-1}\mathbf{D}_{nm}
    \end{align}
    Since $\xi_{nm}$ depends only on $s_n$ and $s_m$, while $g_\ell$ is centered and independent of them for $\ell\notin\{n,m\}$, all terms with $\ell\notin\{n,m\}$ vanish. Therefore,
    \begin{align}
        \sum_{\ell}
        \mathbb{E}\left(\xi_{nm}^2g_{\ell}\right)= \mathbb{E}\left(\xi_{nm}^2g_{n} +  \xi_{nm}^2g_{m}\right).
    \end{align}
    Moreover, for a fixed pair $n<m$, we have
    \begin{align}
         &\nonumber\mathbb{E}\left(\xi_{nm}^2g_{n}\right) \\&= \frac{1}{4}\mathbb{E}\left(\eta_{nm}^2g_{n}\right)+\frac{1}{4}\mathbb{E}\left(\eta_{nm}^{\ast2}g_{n}\right) + \frac{1}{2}\mathbb{E}\left(\left|\eta_{nm}\right|^2g_{n}\right).
    \end{align}
    By noting $\mathbb{E}(s_n^2) = 0$ due to Assumption \ref{assumption2}, $\mathbb{E}\left(\eta_{nm}^2g_{n}\right) = \mathbb{E}\left(\eta_{nm}^{\ast2}g_{n}\right) = 0$, yielding
    \begin{align}
         \nonumber\mathbb{E}\left(\xi_{nm}^2g_{n}\right) &= \frac{1}{2}\mathbb{E}\left(\left|\eta_{nm}\right|^2g_{n}\right) 
         \\
         &\nonumber= \frac{1}{2}\mathbb{E}\left[\left|s_n^\ast[\mathbf{K}]_{n,m}s_m\right|^2\left(\left|s_{n}\right|^2-1\right)\right] \\
         &\nonumber= \frac{1}{2}\left|[\mathbf{K}]_{n,m}\right|^2\mathbb{E}\left(\left|s_m\right|^2\right)\mathbb{E}\left[\left|s_n\right|^2\left(\left|s_n\right|^2-1\right)\right]\\
         & = \frac{1}{2}\left|[\mathbf{K}]_{n,m}\right|^2\left(\kappa - 1\right).
    \end{align}
    As a consequence,
    \begin{align}
        &\nonumber\mathbb{E}\left(
        \mathbf{Z}\dot{\mathbf{Z}}
        \dot{\mathbf{Z}}
        \right)  =\\
        &2\left(\kappa-1\right)\sum_{n<m}\left|[\mathbf{K}]_{n,m}\right|^2\left(\mathbf{A}_{n} +\mathbf{A}_m\right)
        \bar{\mathbf{J}}^{-1}\mathbf{D}_{nm}
        \bar{\mathbf{J}}^{-1}\mathbf{D}_{nm}.
    \end{align}
    Recall $\|\mathbf{D}_{nm}\|_F^2=\Theta\left(|f_n-f_m|^2\right)$, we have
    \begin{align}\label{JD_scale}
        \left\|
        \bar{\mathbf{J}}^{-1}\mathbf{D}_{nm}
        \right\|_2^2
        =
        \Theta\left(N^{-2}|f_n-f_m|^2\right).
    \end{align}
    Observing $\|\mathbf A_n\|_2=O(N^{-1})$, we arrive at
    \begin{align}
        \nonumber
        \left\|
        \mathbb{E}\left(
        \mathbf{Z}\dot{\mathbf{Z}}\dot{\mathbf{Z}}
        \right)
        \right\|_2&=
        O\left(
        N^{-3}
        \sum_{n<m}|[\mathbf K]_{n,m}|^2|f_n-f_m|^2
        \right)\\
        &=
        O\left(N^{-3}C_\mathbf{K}\right),
    \end{align}
    where $C_{\mathbf{K}}$ is defined in \eqref{CK_DEF}.

    We next incorporate the indicator $1_{\Omega_N}$. Since the
    constellation alphabet is finite, $\|\mathbf Z\|_2$, $\|\dot{\mathbf Z}\|_2$,
    and $\|\ddot{\mathbf Z}\|_2$ are polynomially bounded in $N$ at $t=0$.
    Moreover, $\Pr(\Omega_N^c)=O(e^{-cN})$. Hence
    \begin{align}
        \left\|
        \mathbb{E}\left(
        \mathbf{Z}\dot{\mathbf{Z}}\dot{\mathbf{Z}}
        \cdot 1_{\Omega_N}
        \right)
        -
        \mathbb{E}\left(
        \mathbf{Z}\dot{\mathbf{Z}}\dot{\mathbf{Z}}
        \right)
        \right\|_2=
        O(e^{-cN}),
    \end{align}
    for some constant $c>0$. Therefore,
    \begin{equation}
        \left\|
        \mathbb{E}\left[
        \mathbf{Z}\dot{\mathbf{Z}}\dot{\mathbf{Z}}
        \cdot 1_{\Omega_N}
        \right]
        \right\|_2
        =
        O\left(N^{-3}C_\mathbf{K}\right).
    \end{equation}
    This proves that every term of Form~1 satisfies
    \begin{equation}\label{form1_weighted_bound_on_Omega_final}
        \left\|
        \mathbb{E}\left[
        \mathbf Z^a\dot{\mathbf Z}\mathbf Z^b
        \dot{\mathbf Z}\mathbf Z^c
        \cdot 1_{\Omega_N}
        \right]
        \right\|_2
        =
        O(\rho^{r-3}N^{-3}C_{\mathbf K}).
    \end{equation}
    
    We now turn to Form~2. As before, the undifferentiated factors of $\mathbf Z$ contribute powers of $\rho$ on $\Omega_N$. Thus,
    \begin{equation}\label{form1_weighted_bound_on_Omega}
        \left\|
        \mathbb{E}\left[
        \mathbf Z^a\ddot{\mathbf Z}\mathbf Z^b
        \cdot 1_{\Omega_N}
        \right]
        \right\|_2
        \le \rho^{r-3} \left\|
        \mathbb{E}\left[
        \mathbf Z\ddot{\mathbf Z}
        {\mathbf Z}
        \cdot 1_{\Omega_N}
        \right]
        \right\|_2,
    \end{equation}
    and it remains to bound the core term
    $\mathbb{E}(\mathbf Z\ddot{\mathbf Z}\mathbf Z\cdot 1_{\Omega_N})$.
    We first consider the corresponding untruncated expectation. Conditioned on
    $|\mathbf{s}|^2$, the phase-unbalanced terms vanish, and hence
    \begin{align}
        \mathbb{E}\left(
        \ddot g_\ell
        \left|\left|\mathbf{s}\right|^2\right.
        \right)=
        2\sum_{i=1}^N |[\mathbf K]_{\ell,i}|^2 |s_i|^2
        +
        2\operatorname{Re}\left\{
        [\mathbf K^2]_{\ell,\ell}|s_\ell|^2
        \right\}.
    \end{align}
    Note that
    $[\mathbf K^2]_{\ell,\ell}=-\sum_{i=1}^N |[\mathbf K]_{\ell,i}|^2$ for the skew-Hermitian matrix $\mathbf{K}$. Therefore,
    \begin{equation}
        \mathbb{E}\left(
        \ddot g_\ell
        \left|\left|\mathbf{s}\right|^2\right.
        \right)
        =
        2\sum_{i=1}^N |[\mathbf K]_{\ell,i}|^2
        \left(g_i-g_\ell\right).
    \end{equation}
    Since $\ddot{\mathbf Z}=\sum_{\ell=1}^N \ddot g_\ell \mathbf A_\ell$, we have
    \begin{align}
        \nonumber
        \mathbb{E}\left(
        \ddot{\mathbf Z}
        \left|\left|\mathbf{s}\right|^2\right.
        \right)
        &=
        \sum_{\ell=1}^N
        \mathbb{E}\left(
        \ddot g_\ell
        \left|\left|\mathbf{s}\right|^2\right.
        \right)\mathbf A_\ell\\
        &=
        2\sum_{\ell=1}^N\sum_{i=1}^N
        |[\mathbf K]_{\ell,i}|^2
        \left(g_i-g_\ell\right)\mathbf A_\ell .
    \end{align}
    The terms with $i=\ell$ vanish. For each unordered pair $n<m$, the two
    ordered pairs $(\ell,i)=(n,m)$ and $(\ell,i)=(m,n)$ contribute
    \begin{align}
        &\nonumber
        2|[\mathbf K]_{n,m}|^2(g_m-g_n)\mathbf A_n
        +
        2|[\mathbf K]_{m,n}|^2(g_n-g_m)\mathbf A_m\\
        &=
        2|[\mathbf K]_{n,m}|^2(g_m-g_n)
        \left(\mathbf A_n-\mathbf A_m\right),
    \end{align}
    where we used $|[\mathbf K]_{m,n}|^2=|[\mathbf K]_{n,m}|^2$. Since
    $\mathbf A_n-\mathbf A_m=\bar{\mathbf J}^{-1}\mathbf D_{nm}$, it follows that
    \begin{align}
        \nonumber
        \mathbb{E}\left(
        \ddot{\mathbf Z}
        \left|\left|\mathbf{s}\right|^2\right.
        \right)
        &=
        2\sum_{\ell=1}^N\sum_{i=1}^N
        |[\mathbf K]_{\ell,i}|^2
        \left(g_i-g_\ell\right)\mathbf A_\ell\\
        &=
        2\sum_{n<m}|[\mathbf K]_{n,m}|^2
        \left(g_m-g_n\right)
        \bar{\mathbf J}^{-1}\mathbf D_{nm}.
    \end{align}
    Therefore,
    \begin{align}
        \mathbb{E}\left(
        \mathbf Z\ddot{\mathbf Z}\mathbf Z
        \right)=
        2\sum_{n<m}|[\mathbf K]_{n,m}|^2
        \mathbb{E}\left[
        \mathbf Z
        \left(g_m-g_n\right)
        \bar{\mathbf J}^{-1}\mathbf D_{nm}
        \mathbf Z
        \right].
    \end{align}
    For each fixed pair $(n,m)$, expanding
    $\mathbf Z=\sum_{\ell=1}^N g_\ell\mathbf A_\ell$ gives
    \begin{align}
        &\nonumber
        \mathbb{E}\left[
        \mathbf Z
        \left(g_m-g_n\right)
        \bar{\mathbf J}^{-1}\mathbf D_{nm}
        \mathbf Z
        \right]\\
        &=
        \sum_{i,j}
        \mathbb{E}\left[
        g_i\left(g_m-g_n\right)g_j
        \right]
        \mathbf A_i
        \bar{\mathbf J}^{-1}\mathbf D_{nm}
        \mathbf A_j .
    \end{align}
    By independence and $\mathbb{E}(g_\ell)=0$, all terms vanish except those
    with $(i,j)=(n,n)$ and $(i,j)=(m,m)$. Since the random variables
    $\{g_\ell\}_{\ell =1}^N$ are identically distributed, we have
    \begin{align}
        &\nonumber
        \mathbb{E}\left[
        \mathbf Z
        \left(g_m-g_n\right)
        \bar{\mathbf J}^{-1}\mathbf D_{nm}
        \mathbf Z
        \right]\\
        &=
        \mathbb{E}(g_n^3)
        \left(
        \mathbf A_m
        \bar{\mathbf J}^{-1}\mathbf D_{nm}
        \mathbf A_m
        -
        \mathbf A_n
        \bar{\mathbf J}^{-1}\mathbf D_{nm}
        \mathbf A_n
        \right),
    \end{align}
    where $\mathbb{E}(g_n^3) = \mathbb{E}\left(|s_n|^6\right)-3\kappa +2$ is uniformly bounded for finite constellations. The bracketed term has an additional cancellation. Indeed, since
    $\mathbf A_n-\mathbf A_m=\bar{\mathbf J}^{-1}\mathbf D_{nm}$, we have
    \begin{align}
        &\nonumber
        \mathbf A_m
        \bar{\mathbf J}^{-1}\mathbf D_{nm}
        \mathbf A_m
        -
        \mathbf A_n
        \bar{\mathbf J}^{-1}\mathbf D_{nm}
        \mathbf A_n\\
        &=
        -\bar{\mathbf J}^{-1}\mathbf D_{nm}
        \bar{\mathbf J}^{-1}\mathbf D_{nm}
        \mathbf A_m
        -
        \mathbf A_n
        \bar{\mathbf J}^{-1}\mathbf D_{nm}
        \bar{\mathbf J}^{-1}\mathbf D_{nm}.
    \end{align}
    Again, using $\|\mathbf A_n\|_2=O(N^{-1})$ and \eqref{JD_scale}, we obtain
    \begin{align}
        &
        \left\|
        \mathbb{E}\left[
        \mathbf Z
        \left(g_m-g_n\right)
        \bar{\mathbf J}^{-1}\mathbf D_{nm}
        \mathbf Z
        \right]
        \right\|_2=
        O\left(N^{-3}|f_n-f_m|^2\right),
    \end{align}
    Consequently,
    \begin{align}
        \nonumber
        \left\|
        \mathbb{E}\left(
        \mathbf Z\ddot{\mathbf Z}\mathbf Z
        \right)
        \right\|_2&=
        O\left(
        N^{-3}
        \sum_{n<m}|[\mathbf K]_{n,m}|^2|f_n-f_m|^2
        \right)\\
        &=
        O\left(N^{-3}C_{\mathbf K}\right).
    \end{align}

    We now incorporate the indicator $1_{\Omega_N}$. By the same argument used for Form~1, we obtain
    \begin{equation}
        \left\|
        \mathbb{E}\left(
        \mathbf Z\ddot{\mathbf Z}\mathbf Z\cdot 1_{\Omega_N}
        \right)
        -
        \mathbb{E}\left(
        \mathbf Z\ddot{\mathbf Z}\mathbf Z
        \right)
        \right\|_2
        =
        O(e^{-cN}).
    \end{equation}
    Hence,
    \begin{equation}
        \left\|
        \mathbb{E}\left(
        \mathbf Z\ddot{\mathbf Z}\mathbf Z\cdot 1_{\Omega_N}
        \right)
        \right\|_2
        =
        O\left(N^{-3}C_{\mathbf K}\right).
    \end{equation}
    Therefore,
    \begin{equation}\label{form2_weighted_bound_on_Omega_final}
        \left\|
        \mathbb{E}\left[
        \mathbf Z^a\ddot{\mathbf Z}\mathbf Z^b
        \cdot 1_{\Omega_N}
        \right]
        \right\|_2
        =
        O\left(\rho^{r-3}N^{-3}C_{\mathbf K}\right).
    \end{equation}

    It remains to sum all differentiated terms in \eqref{ddot_Z_r}. The second derivative of $\mathbf Z^r(t)$ contains at most $r^2$ terms of Form~1 and $r$ terms of Form~2. Therefore, combining
    \eqref{form1_weighted_bound_on_Omega_final} and
    \eqref{form2_weighted_bound_on_Omega_final} gives
    \begin{equation}\label{Z_power_derivative_weighted_bound}
        \left\|
        \mathbb{E}\left[
        \ddot{\mathbf Z}^{r}(0)\cdot 1_{\Omega_N}
        \right]
        \right\|_2
        =
        O(r^2\rho^{r-3}N^{-3}C_{\mathbf K}),
        \qquad r\ge 3.
    \end{equation}
    Moreover, using
    $\|\bar{\mathbf J}^{-1}\|_2=O(N^{-1})$, each differentiated term satisfies
    \begin{align}\label{R3_each_term_bound}
        &\nonumber
        \left|
        \operatorname{Tr}\left[
        \mathbf T
        \mathbb{E}\left[
        \ddot{\mathbf Z}^{r}(0)\cdot 1_{\Omega_N}
        \right]
        \bar{\mathbf J}^{-1}
        \right]\right|\\
        &\le
        \|\mathbf T\|_F
        \left\|
        \mathbb{E}\left[
        \ddot{\mathbf Z}^{r}(0)\cdot 1_{\Omega_N}
        \right]
        \right\|_2
        \left\|\bar{\mathbf J}^{-1}\right\|_2 
         =
        O(r^2\rho^{r-3}N^{-4}C_{\mathbf K}).
    \end{align}
    Since $\rho<1$, the series $\sum_{r=3}^{\infty}r^2\rho^{r-3}$ is finite.
    Therefore, the differentiated Neumann series in \eqref{R_3_Hessian_expansion} is absolutely summable. Summing over $r$ in \eqref{R_3_Hessian_expansion}, we obtain
    \begin{equation}
       \left.
        \frac{d^2}{dt^2}
        \mathbb{E}\left\{
        \operatorname{Tr}\left[\mathbf T\mathbf R_3(t)\right]
        \cdot 1_{\Omega_N}
        \right\}
        \right|_{t=0}
        =
        O(N^{-4}C_{\mathbf K}).
    \end{equation}

    \subsection{Final Result}
    Combining the above estimates with the resolvent expansion, we obtain
    \begin{align}\label{Hessian_final}
        &\nonumber\left.\frac{d^{2}}{dt^{2}}\mathbb{E}\left\{f\left[\mathbf{U}(t)\right]\right\}\right|_{t=0} =\underbrace{\left.\frac{d^{2}}{dt^{2}}\mathbb{E}\left\{\operatorname{Tr}\left[\mathbf{T}\mathbf{R}_2(t)\right]\right\}\right|_{t=0}}_{\Theta\left(N^{-3}C_{\mathbf{K}}\right)}\\&-\underbrace{\left.\frac{d^{2}}{dt^{2}}\mathbb{E}\left\{\operatorname{Tr}\left[\mathbf{T}\mathbf{R}_3(t)\right]\cdot 1_{\Omega_N}\right\}\right|_{t=0}}_{O\left(N^{-4}C_{\mathbf{K}}\right)} +O(e^{-cN}).
    \end{align}
    The leading term is non-negative and dominates the higher-order term. This proves that the Riemannian Hessian at the OFDM point is non-negative for sufficiently large $N$.

\balance
\bibliographystyle{IEEEtran}
\bibliography{CRB_OFDM_OPT}

\end{document}